%% file: main.tex
\documentclass[11pt]{article}

\usepackage{color}
\usepackage[colorlinks=true,linkcolor=blue,citecolor=blue]{hyperref}
\usepackage{amsmath, amssymb, amsthm}
\usepackage{mathtools}
\usepackage[noabbrev,capitalize,nameinlink]{cleveref}
\crefname{equation}{}{}
\usepackage{fullpage}
\usepackage{graphics}
\usepackage{pifont}
\usepackage{tikz}
\usepackage{bbm}
\usepackage[T1]{fontenc}

\DeclareMathOperator*{\argmin}{arg\,min}
\usetikzlibrary{arrows.meta}
\usepackage{environ}
\usepackage{framed}
\usepackage{url}
\usepackage[linesnumbered,ruled,vlined]{algorithm2e}
\usepackage[noend]{algpseudocode}
\usepackage[labelfont=bf]{caption}
\usepackage{cite}
\usepackage{framed}
\usepackage[framemethod=tikz]{mdframed}
\usepackage{appendix}
\usepackage{graphicx}
\usepackage[textsize=tiny]{todonotes}
\usepackage{tcolorbox}
\usepackage{nicefrac,xfrac}
\allowdisplaybreaks[1]

\newcommand\remove[1]{}

\newtheorem{theorem}{Theorem}
\newtheorem{lemma}{Lemma}[section]
\newtheorem*{lemma*}{Lemma}
\newtheorem{corollary}[lemma]{Corollary}
\newtheorem*{corollary*}{Corollary}

\newtheorem*{theorem*}{Theorem}

\newtheorem{fact}[lemma]{Fact}

\newtheorem{definition}[lemma]{Definition}

\newtheorem{prob}[lemma]{Problem}
\newtheorem*{rem*}{Remark}

\newcommand\R{\mathbb{R}}

\newcommand\E{\mathcal{E}}

\newcommand{\eps}{\varepsilon}
\renewcommand{\O}{\widetilde{O}}
\newcommand{\pe}{\preceq}
\newcommand{\se}{\succeq}

\newcommand{\assign}{\leftarrow}
\newcommand{\otilde}{\O}
\renewcommand{\forall}{\mathrm{\text{ for all }}}

\newcommand{\g}{\nabla}
\newcommand{\diag}[1]{\mathbf{diag}\left(#1\right)}

\newcommand{\new}{\mathrm{new}}

\newcommand{\final}{\mathrm{final}}
\newcommand{\Oracle}{\textsc{Oracle}}
\newcommand{\OracleSmall}{\textsc{OracleSmall}}
\newcommand{\ApproxLewis}{\textsc{ApproxLewis}}
\newcommand{\holder}{H\"{o}lder's inequality}
\newcommand{\ApproxLargeWeights}{\textsc{ApproxLargeWeights}}
\newcommand{\ApproxRegLewis}{\textsc{ApproxRegLewis}}
\newcommand{\ApproxLev}{\textsc{ApproxLev}}
\newcommand{\hDelta}{\widehat{\Delta}}
\newcommand{\hw}{\widehat{w}}
\newcommand{\hmW}{\widehat{\mW}}
\newcommand{\hy}{\widehat{y}}
\newcommand{\T}{\mathcal{T}}
\newcommand{\nnz}{\mathrm{nnz}}
\newcommand{\lse}{\mathrm{lse}}

\newcommand{\ma}{\mathbf{A}}
\newcommand{\md}{\mathbf{D}}
\newcommand{\mw}{\mathbf{W}}

\newcommand{\norm}[1]{\left\lVert#1\right\rVert}
\newcommand{\hess}{\nabla^2}
\newcommand{\xopt}{x_{\star}}
\newcommand{\Eps}{\mathcal{E}}

\newcommand{\mm}{\mathbf{M}}
\newcommand{\OPT}{\mathsf{OPT}}

\foreach \x in {A,...,Z}{%
	\expandafter\xdef\csname m\x\endcsname{\noexpand\mathbf{\x}}
}

\newif\ifrandom
\randomtrue

\newcommand{\defeq}{\stackrel{\mathrm{\scriptscriptstyle def}}{=}}

\newcommand{\poly}{{\mathrm{poly}}}

\newcommand{\todolater}[1]{}
\newcommand{\sign}{\mathrm{sign}}

\crefname{algocf}{Algorithm}{Algorithms}

\begin{document}

\title{Improved Iteration Complexities for \\ Overconstrained $p$-Norm Regression}

\author{Arun Jambulapati \\
Stanford University	\\
\texttt{jmblpati@stanford.edu}
	 \and 
Yang P. Liu \\
Stanford University \\
\texttt{yangpatil@gmail.com}
\and 
Aaron Sidford \\
Stanford University \\
\texttt{sidford@stanford.edu}
}

\begin{titlepage}
\clearpage\maketitle
\thispagestyle{empty}

\input{abstract.tex}

\end{titlepage}

\newpage

\input{intro}
\input{prelim}
\input{largep_ckmst}
\input{largep_ms}
\input{smallq}

\section*{Acknowledgements}
\label{sec:acknowledgements}

We thank Michael B. Cohen, Yair Carmon, Qijia Jiang, Yujia Jin, Yin Tat Lee, Kevin Tian, and Richard Peng for helpful discussions. We also would like to thank anonymous reviewers for several helpful comments in improving the presentation of the paper.

Aaron Sidford was supported in part by a Microsoft Research Faculty Fellowship, NSF CAREER Award CCF-1844855, NSF Grant CCF-1955039, a PayPal research award, and a Sloan Research Fellowship.
Yang P. Liu was supported by the Department of Defense (DoD) through the National Defense Science and Engineering Graduate Fellowship, and NSF CAREER Award CCF-1844855 and NSF Grant CCF-1955039.

{\small
\bibliographystyle{alpha}
\bibliography{refs}}

\appendix

\input{proofs}
\input{linf.tex}

\end{document}

%% file: abstract.tex
\begin{abstract}
In this paper we obtain improved iteration complexities for solving $\ell_p$ regression. We provide methods which given any full-rank $\mathbf{A} \in \mathbb{R}^{n \times d}$ with $n \geq d$, $b \in \mathbb{R}^n$, and $p \geq 2$ solve $\min_{x \in \mathbb{R}^d} \left\|\mathbf{A} x - b\right\|_p$ to high precision in time dominated by that of solving $\widetilde{O}_p(d^{\frac{p-2}{3p-2}})$\footnote{We use $\widetilde{O}_p(\cdot)$ to hide $\log^{O(1)}n$ factors and constants depending only on $p$. In this work, our dependence on $p$ is at most $p^{O(p)}$ for all algorithms, and can in fact be made polynomial in most cases.} linear systems in $\mathbf{A}^\top \mathbf{D} \mathbf{A}$ for positive diagonal matrices $\mathbf{D}$. This improves upon the previous best iteration complexity of $\widetilde{O}_p(n^{\frac{p-2}{3p-2}})$ (Adil, Kyng, Peng, Sachdeva 2019). As a corollary, we obtain an $\widetilde{O}(d^{1/3}\epsilon^{-2/3})$ iteration complexity for approximate $\ell_\infty$ regression. Further, for $q \in (1, 2]$ and dual norm $q = p/(p-1)$ we provide an algorithm that solves $\ell_q$ regression in $\widetilde{O}(d^{\frac{p-2}{2p-2}})$ iterations.

To obtain this result we analyze row reweightings (closely inspired by $\ell_p$-norm Lewis weights) which allow a closer connection between $\ell_2$ and $\ell_p$ regression. We provide adaptations of two different iterative optimization frameworks which leverage this connection and yield our results. The first framework is based on iterative refinement and multiplicative weights based width reduction and the second framework is based on highly smooth acceleration. Both approaches yield $\widetilde{O}_p(d^{\frac{p-2}{3p-2}})$ iteration methods but the second has a polynomial dependence on $p$  (as opposed to the exponential dependence of the first algorithm) and provides a new alternative to the previous state-of-the-art methods for $\ell_p$ regression for large $p$.
\end{abstract}

%% file: intro.tex
\section{Introduction}
\label{sec:intro}

In this paper, we consider the problem of solving \emph{$\ell_p$ regression} for $p \in (1, \infty)$ to high precision.
\begin{definition}[$\ell_p$ regression]
\label{def:lpprob}
Given a full-rank matrix\footnote{We assume throughout that the matrix $\mA$ is full-rank with $n \geq d$ throughout for simplicity, and our results extend directly to the general case, for example by replacing inverses with pseudoinverses.} $\mA \in \R^{n \times d}$, a vector $b \in \R^n$, and a scalar $p \geq 1$ we say that an algorithm solves $\ell_p$ regression to $\eps$-accuracy if it outputs $y \in \R^n$ satisfying
\begin{equation}
\label{eq:lp_regression}
\norm{\mA y - b}_p \le (1+\eps)\min_{x \in \R^d} \norm{\ma x - b}_p.
\end{equation}
We say that such an algorithm is \emph{high precision} if the runtime depends polynomially on $\log(1/\eps)$.
\end{definition}

Beyond possible applications and utility for data analysis (see \cite{MM13,WZ13} and references therein), the problem of $\ell_p$ regression is a prominent testbed for new techniques in optimization and numerical analysis. Varying $p$ causes \eqref{eq:lp_regression} to smoothly interpolate between least squares regression ($p=2$), which can be solved with a single linear system solve, and linear programming ($p\in\{1,\infty\}$) \cite{LS15}, which is only known to be solvable to high precision with $\O(\sqrt{n})$ linear systems via classical interior point methods (IPMs) \cite{Ren88}, and more recently $\O(\sqrt{d})$ linear systems \cite{LS19}.

Interestingly, although \cite{BCLL18} showed that IPMs do not directly yield $o(\sqrt{n})$ iteration complexities for $\ell_p$ regression, there is a line of work \cite{BCLL18,AKPS19,APS19,AS20,ABKS21} which obtained improved iteration complexities via alternative methods; the current state-of-the-art iteration complexity for $p \ge 2$ is $\O_p(n^{\nicefrac{(p-2)}{(3p-2)}})$. These improvements touch on a range of advanced optimization techniques including homotopy methods, iterative refinement, high-order acceleration, and width-reduction. This line of work is closely related to work which solves approximate $\ell_\infty$ regression in $\O(n^{1/3}\eps^{-O(1)})$ iterations \cite{CMMP13} where again, improvements and simplifications have been achieved through multiple techniques \cite{DLS18,EV19,CJJJLST20}. Additionally, work on $\ell_p$ regression for structured graph incidence matrices $\mA$ \cite{KPSW19,ABKS21} has led to improved running times for unit capacity maxflow, bipartite matching, and mincost flows \cite{LS20a,KLS20,AMV20}.

Given this progress, a natural open problem is to bridge the gap between the known iteration complexities for $\ell_p$ regression and the $\tilde{O}(\sqrt{d})$ bound achievable by IPMs for $\ell_1$ and $\ell_\infty$ regression \cite{LS19} by providing an iteration complexity for $\ell_p$ regression that depends on $d$ as opposed to $n$. Additionally, the relationship between various techniques for achieving these iteration complexities, especially acceleration and width-reduction, remains somewhat mysterious (see \cite{ABS21} for further discussion on this relationship). Consequently, understanding the complexity of $\ell_p$ regression is fundamental for advancing and clarifying the power of various optimization techniques.

In this paper we take steps to address these questions and improve the complexity for solving $\ell_p$ regression. Our main result is a pair of algorithms, based respectively on the iterative refinement framework and width reduction techniques \cite{AKPS19}, and the Monteiro-Svaiter/highly-smooth acceleration framework \cite{BJLLS19,CJJJLST20}, each of which, for $p \geq 2$ solve $\ell_p$ regression with $\O_p(d^{\nicefrac{(p-2)}{(3p-2)}})$ linear system solves. This improves an $n$ to a $d$ in the iteration dependencies for the state-of-the-art methods for $\ell_p$ regression.

The key notion used in our methods are reweightings of $\mA$ closely related to Lewis weights \cite{Lewis78,CP15}, which allow for a closer relationship between the $\ell_2$ and $\ell_p$ norms induced by $\mA$ (\cref{lemma:weightkey}). We directly leverage this connection induced by $\ell_p$-norm Lewis weights in the context of the optimization frameworks discussed previously to achieve our result, as opposed to previous works on (approximate) $\ell_p$ regression that used $\ell_p$ Lewis weights to construct $\ell_p$-norm sparsifiers or subspace embeddings \cite{DDHR09,CW13,WZ13,MM13b,CDMMMW16,CWW19}. As a result, these previous iteration complexities for $\ell_p$ regression had iteration complexities of the form $d^{\Omega(p)}$, while ours is always $\O(d^{1/3})$, even for large $p = \O(1)$.

\subsection{Our Results}

Here we state the main results of our paper. As is the case with several results on regression \cite{CMMP13,BCLL18,AKPS19,EV19,CJJJLST20}, the primary subroutine used by our algorithms is a linear system solver for $\mA^\top\mD\mA$ for positive diagonal matrices $\mD$. We focus on bounding the number of iterations or calls to such a linear system solver in our algorithms. Accordingly, let $\T_\mA$ denote the time for solving a linear system in $\mA^\top\mD\mA$ for positive diagonal matrices $\mD$.\footnote{Throughout we assume that all solves to $\mA^\top\mD\mA$ are exact. Typically it suffices to set the solver error to be polynomially small in $n, d$ and the largest entries of the input vectors and matrix $\mA$. This increases the running time of the solver by polylogarithmic factors.} We also use ``with high probability" (whp.) throughout to mean with probability at least $1 - n^{-C}$ for any constant $C$.

In this work, we focus on presenting iteration complexity improvements for $\ell_p$ regression problems. We choose to focus on iteration complexity improvements in the work as opposed to runtimes for the sake of achieving a cleaner and simpler presentation. We elaborate on this in the final paragraph in previous works (\cref{sec:related}).
\begin{theorem}[High precision $\ell_p$ regression for $p \ge 2$]
\label{thm:largemain}
There is an algorithm that given any $\mA \in \R^{n \times d}$, $b \in \R^n$, $p \ge 2$ whp. returns an $\eps$-accuracy solution to $\ell_p$ regression in time $\O_p(d^\frac{p-2}{3p-2} \cdot \T_\mA)$.
\end{theorem}
Here the $\O_p(\cdot)$ hides $\poly(\log n, \log(1/\eps))$ factors and constants depending on $p$ (at worst exponential). As a corollary we also obtain high precision solvers for the Lagrange dual problem $\min_{\mA^\top x = b} \|x\|_q$ for $q = p/(p-1)$ in $\O_p(d^\frac{p-2}{3p-2}\T_\mA)$ time whp. This improves over the $\O_p(n^\frac{p-2}{3p-2})$ iteration complexity of \cite{AKPS19} for any tall matrix $\mA$.

Similar ideas as those used to show \cref{thm:largemain} can be used to give improved iteration complexities for approximate $\ell_\infty$ regression, which we show in \cref{sec:appendixinf}.
\begin{theorem}[Approximate $\ell_\infty$ regression]
\label{thm:linfmain}
There is an algorithm that given any $\mA \in \R^{n \times d}$ and $b \in \R^n$, whp. computes an $\eps$-accurate solution to $\ell_\infty$ regression in time $\O(d^{1/3}\eps^{-2/3} \cdot \T_\mA)$.
\end{theorem}
This improves over the $O(n^{1/3}\eps^{-O(1)})$ iteration bounds achieved by \cite{CMMP13,EV19,CJJJLST20}.

We also obtain improved results for $\ell_q$ regression for $q \le 2$ for sufficiently tall matrices $\mA$.
\begin{theorem}[High precision $\ell_q$ regression for $q \le 2$]
\label{thm:smallmain}
There is an algorithm that given any $\mA \in \R^{n \times d}$, $b \in \R^n$, and $q \in (1,2]$ whp. returns an $\eps$-accurate solution to $\ell_q$ regression
in time $\O_p(d^\frac{p-2}{2p-2}  \cdot \T_\mA)$ for $p = q/(q-1)$.
\end{theorem}
As a corollary we also get high accuracy solvers for the Lagrange dual problem $\min_{\mA^\top x = b} \|x\|_p$ for $p = q/(q-1)$ in $\O_p(d^\frac{p-2}{2p-2}\T_\mA)$ time whp.
This improves over the $\O_p(n^\frac{p-2}{3p-2})$ iteration complexity of \cite{AKPS19} for any sufficiently tall matrix $\mA$ with $n = \omega(d^\frac{3p-2}{2p-2}).$

\subsection{Previous Work}
\label{sec:related}
Here we briefly survey related work to the problems we consider and the optimization and numerical methods we build upon.

\paragraph{Regression:} Beyond the works mentioned earlier, there are several results on $\ell_1$ or $\ell_\infty$ regression in both the low precision ($\poly(1/\eps)$ dependence) \cite{Clarkson05,Nesterov09,YCRM16,DLS18}, and high precision regimes corresponding to linear programming \cite{MM13}. Additionally the works \cite{CJJJLST20,ABS21,CJJS21} study the more general problem of quasi-self-concordant optimization, which captures $\ell_p$ regression as well as logistic regression \cite{LWK08,Bach2010}. There are several results on $\ell_p$ regression that are based on sparsification or subspace embeddings using (variants of) $\ell_p$ Lewis weights \cite{DDHR09,CW13,WZ13,MM13b,CDMMMW16,CWW19}. While our result also uses a variant of $\ell_p$ Lewis weights, we do not sparsify. This is key to achieving our iteration complexities because sparsification of an $\ell_p$ norm objective for $p \ge 2$ requires at least $\Omega(d^{p/2})$ rows \cite{CP15}. This leads to iteration complexities of at least $(d^{p/2})^{\nicefrac{(p-2)}{(3p-2)}} = d^{\nicefrac{(p^2-2p)}{(6p-4)}}$ and runtimes of $\O(\nnz(\mA) + d^{\Omega(p)})$, as was noted in \cite[Theorem 2.6]{ABKS21}. On the other hand, our iteration complexity is at most $\O(d^{1/3})$, independent of $p$. This allows us to achieve a $\O(d^{1/3}\eps^{-2/3})$ iteration complexity for $\ell_\infty$-regression in \cref{thm:linfmain}, while the aforementioned works using sparsification are unable to. Very recently, \cite{GPV21} achieved a $O(n^{\theta})$ runtime for $\ell_p$ norm regression on sufficiently sparse matrices $\mA$ for some $\theta < \omega$ (the matrix multiplication exponent). For $p$ near $2$, they improved this to $\widetilde{O}(\nnz(\mA) + d^{\theta}).$

\paragraph{High-order acceleration:} Our Monteiro-Svaiter acceleration algorithm builds upon works pertaining to the acceleration of functions with Lipschitz $p$-th order derivatives. For $p = 1$ this corresponds to classic acceleration of smooth functions that attains error $\O(1/k^2)$ over $k$ iterations \cite{Nesterov83}. A series of works \cite{MS13,AH18,ASS19,Nesterov19,BP19,BJLLS19, GDGVSU0W19,CJJJLST20} has shown that the optimal error bound is given by $\O(1/k^{\nicefrac{(3p+1)}{2}})$ over $k$ iterations for functions with Lipschitz $p$-th order derivatives. Our Monteiro-Svaiter acceleration algorithm for $\ell_p$ regression directly utilizes a generalized accelerated proximal-point framework from \cite{BJLLS19}. Additionally, \cite{Bullins18} has given an algorithm that achieves an accelerated convergence rate for the more general problem of minimizing structured convex quartics which captures $\ell_4$ regression but has an additional third order tensor term. It is interesting to understand whether our methods extend to that setting.

\paragraph{Width reduction:} In addition to its applications for regression problems as described, similar width reduction techniques have been applied to give improved runtimes for the maxflow problem in both approximate regimes \cite{CKMST11,KMP12} and in unit capacity graphs \cite{Mad13,Mad16,LS20a,LS20b,Kat20,CMSV17,AMV20}. Additionally, the Iteratively Reweighted Least Squares (IRLS) algorithm of Ene-Vladu \cite{EV19} gives an alternate approach based on width reduction for achieving a $\O(n^{1/3}\eps^{-2/3})$ iteration complexities for $\ell_1$ and $\ell_\infty$ regression, matching the iteration complexity of \cite{CJJJLST20}. We believe that applying ideas from the analysis of \cite{EV19} can potentially be used to simplify our width reduction algorithm for $\ell_p$ regression given in \cref{sec:largep_ckmst}.

\paragraph{Runtime improvements for regression problems.} We briefly discuss why we focus on presenting iteration complexity improvements in this work, as opposed to runtimes for $\ell_p$ regression. In general, obtaining improving runtimes for regression problems beyond improving the iteration complexity has been through \emph{inverse maintenance} techniques \cite{V89,V90,LS15}, and more recently heavy hitter and iterate maintenance \cite{CLS19,LSZ19,B20,BLSS20,BLNPSSSW20,BLLSSSW21}, to speed up the amortized time to solve the linear systems in $\mA^\top\mD\mA$ and implicitly maintain the iterates. This direction has seen an explosion of work recently, with the state-of-the-art runtimes for solving linear programs (eg. high precision $\ell_1$ regression) being some combination of the recent works $\O(n^{\max\{\omega, 2+1/18\}})$ \cite{JSWZ21}, $\O(nd + d^{2.5})$ \cite{BLLSSSW21}, and $\O(\nnz(\mA)d^{0.5} + d^{2.5})$ \cite{LS15}. The authors believe that all our improved iteration complexities in \cref{thm:largemain,thm:linfmain,thm:smallmain} can be combined with ideas from the aforementioned works to achieve concrete runtime improvements for $\ell_p$ regression. However, given the rapidly evolving progress in inverse maintenance and relative complexity of the methods, we choose to focus on iteration complexities in this work to give a cleaner and simpler presentation of our ideas.

\subsection{Our Approach}

Here we focus on presenting our approach for $p \ge 2$ (\cref{thm:largemain}) and briefly describe our approach for $q \le 2$ (\cref{thm:smallmain}). Both of our algorithmic frameworks (width reduction and acceleration) are based on leveraging properties of $\ell_p$ Lewis weights. While $\ell_p$ Lewis weights have been used in several previous results on $\ell_p$ regression (as described in \cref{sec:related}), these works primarily used Lewis weights to construct sparsifiers or subspace embeddings. We take a different perspective, and instead leverage a key fact about approximate $\ell_p$ Lewis weights that they provide an ellipse which approximates the $\|\mA x\|_p$. This has appeared in \cite[pg.115]{Banach} and \cite[Lemma 3.6]{CWW19}. Precisely, if $w \in \R^n$ are the $\ell_p$-Lewis weights for $\mA$ then
\begin{equation}
\|\mA x\|_p \le \|\mW^{\frac12-\frac1p}\mA x\|_2 \le \|w\|_1^{\frac12-\frac1p}\|\mA x\|_p \forall x \in \R^n. \label{eq:upperp}
\end{equation}
The lower bound follows from the definition of $\ell_p$ Lewis weights, which we give a self-contained proof of in \cref{lemma:weightkey}, and the upper bound follows from \holder. Because $\|w\|_1 = d$ for the $\ell_p$ Lewis weights, the distortion between the lower and upper bounds in \eqref{eq:upperp} is $d^{1/2-1/p}$, which leads to $d$-dependent iteration complexities. While it is not known how to exactly compute the $\ell_p$ Lewis weights for $p \ge 4$ in $\O(\T_\mA)$ time we are still able to argue that we can efficiently compute weights $w$ satisfying \eqref{eq:upperp} but $\|w\|_1 \le 2d$ (\cref{lemma:approxplarge}). This is done by mimicking an argument of \cite{CCLY19} for the $p = \infty$ case which corresponds to computing an approximate John ellipse.

We show that it is possible to leverage our perspective on \eqref{eq:upperp} within either iterative refinement \cite{AKPS19} or an acceleration framework (based on the acceleration framework of \cite{BJLLS19}). While these frameworks are largely compatible with \eqref{eq:upperp}, there are notable conceptual differences which we now discuss. In the iterative refinement framework, the problem of $\ell_p$ regression is reduced to approximately minimize problems that are a combination of a linear term, $\ell_p$ norm term, and $\ell_2$ regularization term (\cref{prob:scaledresidual}). As in \cite{AKPS19}, we use a width-reduced multiplicative weights update (MWU) to reduce the iteration complexity. The main difference is that we show an energy boosting lemma in the width-reduced MWU (\cref{lemma:energyincrease}) that allows for resistances to more than double (while still providing significant increase to the energy potential) by leveraging stability from \eqref{eq:upperp}, while in standard energy boosting the energy does not increase significantly beyond resistances increasing by a constant factor. While the proof follows gracefully from low-rank update formulas, we believe that this is an interesting conceptual point. Our second acceleration-based algorithm repeatedly solves proximal subproblems of the form $\min_x \norm{\ma x - b}_p^p + O(p)^p \norm{x-y}_{\ma^\top \mw^{1-2/p} \ma}^p$: we show that such regularized problems may be solved efficiently by using stability given by \eqref{eq:upperp}. Interestingly, a more na\"{i}ve application of acceleration of ball-constrained Newton methods \cite{CJJJLST20} leads to an iteration complexity of $\O_p(d^{1/3})$. However, our acceleration method achieves an iteration complexity of $\O_p(d^{\nicefrac{(p-2)}/({3p-2})})$ and provides an acceleration-based alternative matching the iteration complexities achieved by width-reduction for intermediate values of $p \in [2, \infty)$.

For the case $q \le 2$ instead of solving $\min_{x \in \R^d} \norm{\mA x - b}_q$ we solve the dual problem
\[ \min_{\substack{\mA^\top x = 0, b^\top x= -1}} \|x\|_p \] 
where $p = q/(q-1)$ is the dual norm. In this setting we also wish to use $\ell_q$ Lewis weights. However the presence of the $\ell_2$ regularizer induced from iterative refinement or the acceleration framework forces us to use a more complex \emph{regularized Lewis weight}, defined in \cref{def:reglewis} (such a concept was also used in \cite[Definition 4.4]{BLLSSSW21}). Unfortunately it seems that this type of regularized Lewis weight is not immediately compatible with the width reduction or acceleration type speedups, and we only achieve a $\O_p(d^{\nicefrac{(p-2)}{(2p-2)}})$ iteration complexity as a result. Consequently, we believe that achieving a matching $\O_p(d^{\nicefrac{(p-2)}{(3p-2)}})$ complexity for the case $q \le 2$ is an important open problem.

\subsection{Paper Organization}

The remainder of the paper is structured as follows. In \cref{sec:prelim} we give preliminaries for our algorithms, e.g. leverage scores, Lewis weights, and iterative refinement. In \cref{sec:largep_ckmst} we provide an iterative refinement and width reduction framework for showing \cref{thm:largemain}. In \cref{sec:largep_ms} we give an alternate approach for the previous result based on the high-order acceleration framework of \cite{BJLLS19}. In \cref{sec:smallq} we show \cref{thm:smallmain} which achieves $d$-dependent (as opposed to $n$-dependent) iteration complexities for $\ell_q$ regression for $q \le 2$.
Finally we show several facts about the computation of (approximate) Lewis weights and iterative refinement in \cref{sec:proofs,proofs:cor12} and our result on approximate $\ell_\infty$ regression (\cref{thm:linfmain}) in \cref{sec:appendixinf}.

%% file: prelim.tex
\section{Preliminaries}
\label{sec:prelim}

\subsection{General Notation}
We use lowercase for vectors, and capital boldface for matrices. We let $\vec{0}, \vec{1}$ denote the all $0, 1$ vectors respectively. Additionally, for a vector the matrix with corresponding capital letter is the diagonal matrix. Throughout we let $w$ denote a weight vector, $r$ denote a positive vector, and $\mW = \diag{w}$ and $\mR = \diag{r}$. We say that a matrix $\mB \in \R^{n \times n}$ is positive semidefinite (PSD) if $x^\top \mB x \ge 0$ for all $x \in \R^n.$ We say that matrices $\mA \pe \mB$ if $\mB - \mA$ is PSD.
For PSD matrices $\mA, \mB$ we say that $\mA \approx_\alpha \mB$ for $\alpha \ge 1$ if $\alpha^{-1}\mB \pe \mA \pe \alpha\mB$.
For a PSD matrix $\mB$ we define the seminorm induced by $\mB$ as $\|x\|_\mB := \sqrt{x^\top \mB x}$.

\subsection{Lewis Weights}

We start by defining the leverage scores and $\ell_p$ Lewis weights of a matrix $\mA$. These are measures of importance of rows of a matrix $\mA$.
\begin{definition}[Leverage scores]
\label{def:levscores}
For a matrix $\mA \in \R^{n \times d}$ whose $i$-th row is the vector $a_i$, the \emph{leverage scores} are given by $\sigma(\mA)_i := a_i^\top(\mA^\top \mA)^{-1} a_i$ for $i \in [n]$.
\end{definition}
It is known that $\sum_{i \in [n]} \sigma(\mA)_i = \mathrm{rank}(\mA)$. Further, the leverage score of the $i$-th row of a matrix $\mA$ is given by the maximum of $|(\mA x)_i|$ over all vectors $x \in \R^d$ satisfying $\|\mA x\|_2 \le 1.$ This provides a concrete way that the leverage scores are $\ell_2$ importance measures of rows.

\begin{fact}[Leverage scores as $\ell_2$ importance]
\label{fact:lev}
For a matrix $\mA \in \R^{n \times d}$ the leverage score of row $i \in [n]$ is given by
\[ \sigma(\mA)_i = \max_{x \in \R^d : \mA x \neq 0} \frac{(\mA x)_i^2}{\|\mA x\|_2^2}. \]
\end{fact}

Lewis weights are a generalization of leverage scores to $\ell_p$ norms for $p \neq 2$.

\begin{definition}[$\ell_p$ Lewis weights]
\label{def:lewisweights}
For a matrix $\mA \in \R^{n \times d}$, the $\ell_p$ \emph{Lewis weights} are given by the unique vector $w \in \R_{\ge0}^n$ satisfying $w_i = \sigma(\mW^{1/2-1/p}\mA)_i$  $\forall$ $i \in [n]$.
\end{definition}
\cite{CP15} proves the existence and uniqueness of $\ell_p$ Lewis weights for all $p \in (0, \infty)$. Additionally, they provide an efficient contractive procedure for approximately computing the $\ell_p$ Lewis weight for $p < 4$. For our applications for $p < 2$, we use a regularized version of these weights, and defer the full statement of the approximation result needed until \cref{lemma:reglevapprox} in \cref{sec:smallq}. For our applications for $p \ge 4$ we show that it is possible to compute weights satisfying the weaker guarantee \eqref{eq:upperp}.

\begin{definition}[$\ell_p$ Lewis weight overestimates]
\label{def:overestimate}
For a matrix $\mA \in \R^{n \times d}$ we say that $w \in \R_{\ge0}^n$ are \emph{$\ell_p$ Lewis weight overestimates} if $d \le \|w\|_1 \le 2d$ and
$w_i \ge \sigma(\mW^{1/2-1/p}\mA)_i$ $\forall i \in [n]$.
\end{definition}
The factor of $2$ is somewhat arbitrary -- any constant factor suffices for our algorithms. In \cref{proofs:approxplarge} prove the following lemma showing that Lewis weight overestimates can be computed with a few linear system solves. Our approach is an extension of that in \cite{CCLY19} which provided a procedure for approximately computing the John ellipse, i.e. the $p = \infty$ case. 

\begin{lemma}[Computing $\ell_p$ Lewis weight overestimates]
\label{lemma:approxplarge}
Given any $\mA \in \R^{n \times d}$ and $p \ge 2$, $\ApproxLargeWeights(\mA, p)$ (\Cref{alg:approx_large})  in $\O(\T_\mA)$ time computes $\ell_p$ Lewis weight overestimates (\cref{def:overestimate}) whp.
\end{lemma}

We can show \eqref{eq:upperp} holds for any $\ell_p$ Lewis weight overestimates.
\begin{lemma}
\label{lemma:weightkey}
For a matrix $\mA \in \R^{n \times d}$ and $\ell_p$ Lewis weight overestimates $w \in \R^n_{>0}$ (\cref{def:overestimate}) we have that $\|\mA x\|_p \le \|\mW^{\frac{1}{2} - \frac{1}{p}}\mA x\|_2$ for all $x \in \R^n$.
\end{lemma}
\begin{proof}
By \cref{fact:lev} we know that
\[ |(\mA x)_i| = w_i^{-\frac12+\frac1p}|(\mW^{\frac12-\frac1p}\mA x)_i| \le w_i^{-\frac12+\frac1p}\sigma(\mW^{\frac12-\frac1p}\mA)_i^{1/2}\|\mW^{\frac12-\frac1p}\mA x\|_2 \le w_i^{1/p}\|\mW^{\frac12-\frac1p}\mA x\|_2. \]
Hence, we see that
\[ \|\mA x\|_p^p = \sum_{i \in [n]} |(\mA x)_i|^p \le \sum_{i\in[n]} w_i^{\frac{p-2}{p}}\|\mW^{\frac12-\frac1p}\mA x\|_2^{p-2}(\mA x)_i^2 = \|\mW^{\frac12-\frac1p}\mA x\|_2^p. \]
Taking the $p$-th root of both sides gives us the result.
\end{proof}

\subsection{Iterative refinement}
\label{sec:iterrefine}

At a high level, the iterative refinement framework for $\ell_p$ norms introduced by \cite{AKPS19} shows that the Bregman divergence of the $\ell_p$ norm, i.e. the function $f(x) = |x|^p$, can be efficiently approximated by an $\ell_2$ and $\ell_p$ component. Using this, we can reduce solving high accuracy $\ell_p$-norm problems to solving approximate $\ell_2$-$\ell_p$ norm problems.
\begin{lemma}[\!\!{\cite[Lemma B.1]{APS19}}]
\label{lemma:refine}
For $x, \Delta \in \R^n$ and $p \ge 2$, we have for $g, r \in \R^n$ defined by $g_i = p|x_i|^{p-2}x_i$ and $r_i = |x_i|^{p-2}$ for $i \in [n]$ that
\begin{align}
\frac{p}{8}\sum_{i \in [n]} r_i \Delta_i^2 + 2^{-p-1}\|\Delta\|_p^p \le \|x+\Delta\|_p^p - \|x\|_p^p - g^\top\Delta \le 2p^2\sum_{i \in [n]} r_i \Delta_i^2 + p^p\|\Delta\|_p^p. \label{eq:refine}
\end{align}
\end{lemma}
There are several more restrictive variations of \cref{lemma:refine} for positive scalars that we use (shown in \cref{proofs:cor12}).
\begin{lemma}
\label{lemma:cor1}
For all $a, b \ge 0$ and $k \ge 2$ we have that $(a+b)^k - a^k \le 3ka^{k-1}b + 3k^k b^k.$
\end{lemma}
The second corollary is useful in slightly different regimes of the exponent $k$.
\begin{lemma}
\label{lemma:cor2}
For all $a, b \ge 0$ and $k \ge 1$ we have that $(a+b)^k - a^k \le 4^k(a^{k-1}b + b^k).$
\end{lemma}

Searching over the value of $g^\top\Delta$ reduces solving $\ell_p$ regression to high accuracy to approximately solving constrained $\ell_2$-$\ell_p$ problems. We call a procedure for approximately solving constrained $\ell_2$-$\ell_p$ problems a \emph{$\gamma$-solver} and provide this reduction, \cite[Theorem 3.1]{APS19} below.
\begin{definition}[$\gamma$-solver]
\label{def:gammasolver}
We call an algorithm a \emph{$\gamma$-solver} if given $\nu \ge 0$, $g \in \R^n$, and a positive diagonal matrix $\mR \in \R^{n\times n}$, it is the case that for  
\[ 
\OPT = \min_{\mC \Delta = 0 ,  g^\top \Delta = -\nu} \Delta^\top\mA^\top\mR\mA\Delta + \|\mA\Delta\|_p^p, 
\] 
the algorithm returns a $\hDelta$ satisfying $\mC\hDelta = 0$, $g^\top\hDelta = -\nu$, and
\[ 
\hDelta^\top\mA^\top\mR\mA\hDelta \le \gamma \OPT \enspace \text{ and } \enspace \|\mA\hDelta\|_p^p \le \gamma^{p-1} \OPT.
 \]
\end{definition}

\begin{lemma}[\!\!{\cite[Theorem 3.1]{APS19}}]
\label{lemma:gammasolver}
Given $\mU \in \R^{n_1 \times d}, \mA \in \R^{n_2 \times d}$ and $b, v$ and $p \ge 2$, we can compute an $x \in \R^d$ satisfying $\mU x = v$ and
\[ \|\mA x - b\|_p \le (1+\eps) \min_{\mU x = v} \|\mA x - b\|_p \] in $O(p^{3.5}\gamma\log(m/\eps))$ calls to a $\gamma$-solver (\cref{def:gammasolver}).
\end{lemma}

%% file: largep_ckmst.tex
\section{Energy Boosting Algorithm for Large $p$}
\label{sec:largep_ckmst}

The goal of this section is to give an algorithm to show \cref{thm:largemain}. By \cref{lemma:gammasolver} and scaling we may assume that $\nu = 1$ and $\OPT = 1$ \cite[Lemma 5.4]{AKPS19} throughout, and we use the following setup throughout the section.
\begin{prob}[Scaled residual]
\label{prob:scaledresidual}
In the \emph{scaled residual problem} we are given $\mA \in \R^{n \times d}$, $g \in \R^d$, and diagonal $\mR \in \R_{\ge0}^{n \times n}$ such that there exists $\xopt \in \R^d$ satisfying $g^{\top}\xopt=-1$ with $\xopt^{\top} \mA^{\top} \mR \mA \xopt \le 1$ and $\|\mA \xopt\|_p \le 1$. We call $y$ an \emph{$\alpha$-approximate solution} to the problem if $g^\top y = -1$, $y^\top \mA^{\top} \mR \mA \xopt \le \alpha$ and  $\|\mA y\|_p \le \alpha$.
\end{prob}

In the notation of \Cref{prob:scaledresidual}, proving the following lemma suffices to show \cref{thm:largemain}.
\begin{lemma}
\label{lemma:largemain}
Given an instance of \cref{prob:scaledresidual}, Algorithm $\Oracle(\mA, g, \mR, p)$ returns an $O(p)^p$-approximate $y$ in $O(p)^pd^\frac{p-2}{3p-2} \cdot \T_\mA$ time whp.
\end{lemma}
\begin{proof}[Proof of \cref{thm:largemain}]
\cref{lemma:largemain} satisfies the conditions of \cref{lemma:gammasolver} for $\gamma = O(p)^p.$ Each call to \cref{lemma:largemain} requires $O(p)^p d^\frac{p-2}{3p-2}$ calls to a solver to $\mA^\top\mD\mA$ so the total number of iterations is \[ \gamma \cdot p^{3.5} \cdot O(p)^p d^\frac{p-2}{3p-2}\log(m/\eps) = O(p)^pd^\frac{p-2}{3p-2}\log(m/\eps). \]
\end{proof}

To show \cref{lemma:largemain} we use the following \cref{algo:pnormbasic}.
\begin{algorithm}[ht]
\caption{$\Oracle(\mA, g, \mR, p)$. Given $\mA, \mR, g$ satisfying \cref{prob:scaledresidual}, returns a $y \in \R^d$ with $g^{\top}y=-1$, $y^{\top}\mA^{\top}\mR \mA y \le O_p(1)$, and $\|\mA y\|_p = O_p(1)$ in $O_p(d^\frac{p-2}{3p-2})$ iterations. \label{algo:oracle}}
$w \assign \ApproxLewis(\mA, p)$. \label{line:approxlewislarge} \Comment{Compute $\ell_p$ Lewis weight overesimates of $\mA$ via \cref{lemma:approxplarge}} \\
$y \assign 0$ and $s \assign w^{1/p}$. \Comment{Iterates} \label{line:inits} \\
$\kappa \assign \kappa_pd^{1/p}$, $\alpha \assign \alpha_pd^{-\frac{p^2-5p+2}{p(3p-2)}}$, $\tau \assign \tau_pd^{\frac{(p-2)(p-1)}{3p-2}}$. \Comment{Constants $\kappa_p, \tau_p$ large, $\alpha_p$ small}\footnote{We show in the proof of \cref{thm:largemain} at the end of this section that $\tau_p = 40^p, \alpha_p = 1/(1000p)$, and $\kappa_p = p$ works.} \\
\For{$t = 0, 1, \ldots , T \defeq \lfloor \alpha^{-1}d^{1/p} \rfloor$}{
	$z \assign \argmin_{g^{\top}x = -1} x^{\top}\mA^{\top}\left(d^{1-\frac2p}\mR + \mS^{p-2}\right)\mA x$. \Comment{$\mS = \diag{s}$} \label{line:setz} \\
	\While{$\|\mA z\|_p^p \ge \tau$ \label{line:iftau}}{
		$S \assign \{ i \in [n] : s_i \le 2^{-\frac{p}{p-2}}\kappa|(\mA z)_i| \}.$ \Comment{Boosting step.} \label{line:chooseS} \\
		$s_i \assign \left(s_i^{p-2} + \frac{\tau^{2/p}|(\mA z)_i|^{p-2}}{4\|\mA z\|_p^p}\right)^\frac{1}{p-2}$ for $i \in S$. \label{line:boost} \\
		$z \assign \argmin_{g^{\top}x = -1} x^{\top}\mA^{\top}\left(d^{1-\frac2p}\mR + \mS^{p-2}\right)\mA x$.
	}
	$y \assign y + \alpha z$ and $s \assign s + \alpha|\mA z|$. \Comment{Progress step.} \label{line:progresss}
}
\Return $(\alpha T)^{-1} y.$
\label{algo:pnormbasic}
\end{algorithm}
It follows the multiplicative weights and width reduction approach of \cite{CKMST11,AKPS19}. The algorithm solves $\ell_2$-norm problems in sequence. When the $\ell_p$ norm of the resulting solution is small enough, i.e. $\|\mA z\|_p^p \le \tau$, the algorithm performs a \emph{progress step} in line \ref{line:progresss}, and adds $z$ to the output. However, whenever the $\ell_p$ norm of the returned solution is large, the algorithm performs a \emph{boosting step} in line \ref{line:boost}, and increases the resistance of the large coordinates contributing significantly to the $\ell_p$ norm $\|\mA z\|_p^p$ to force them to be smaller in future iterations.

To analyze Algorithm $\Oracle(\mA, g, \mR, p)$ in \cref{algo:pnormbasic} and thereby prove \cref{lemma:largemain}, we analyze two potential functions following the approach and notation of \cite{AKPS19}. The first is $\Phi(s) \defeq \|s\|_p^p$, and the second is the \emph{energy} (where $\mS \defeq \diag{s}$)
\[ 
\E(s) \defeq \min_{g^{\top}x=-1} x^\top\mA^\top\left(d^{1-\frac2p}\mR + \mS^{p-2}\right)\mA x. 
\] 
We show that a progress or boosting step doesn't increase $\Phi$ by too much, and that a boosting step significantly increases the energy. Combining this with an energy upper bound completes the proof. To reason about the energy increase we use the following alternate characterization of the energy.
\begin{lemma}
\label{lemma:energysolve}
For any symmetric positive definite matrix $\mB$ and vector $g$ we have
\begin{equation}
\label{eq:quad}
\argmin_{g^\top x = -1} x^\top \mB x = -\frac{1}{g^\top \mB^{-1} g} \mB^{-1} g
\text{ and } 
\min_{g^\top x = -1} x^\top \mB x = (g^\top \mB^{-1} g)^{-1}\,.
\end{equation}
\end{lemma}
\begin{proof}
Let $x^*$ be the minimizer of \eqref{eq:quad}. Note that $g^\top x^* = -1$ by assumption and $\mB x^* = \alpha^* g$ for some unknown $\alpha^*$. Consequently, $x^* = \alpha^* \mB^{-1} g$ and the claim follows from
\[
-1 = g^\top x^* = \alpha^* g^\top \mB^{-1} g.
\]
The second claim follows by using this value to compute ${x^*}^\top \mB x^*.$
\end{proof}

\begin{lemma}[Energy upper bound]
\label{lemma:energyupper}
In \cref{prob:scaledresidual}, for any vector $s$ satisfying $s \ge w^{1/p}$ coordinate-wise for $\ell_p$ Lewis weight overestimates $w$ (\cref{def:overestimate}), we have $\E(s) \le 2\Phi(s)^{1-\frac2p}.$
\end{lemma}
\begin{proof}
Let $\xopt$ be as in \cref{prob:scaledresidual}. By \holder~we have that
\begin{align*}
\E(s) &\le \xopt^\top\mA^\top\left(d^{1-\frac2p}\mR + \mS^{p-2}\right)\mA \xopt \le d^{1-\frac2p} + \|\mA \xopt\|_p^2 \|s\|_p^{p-2} \\ &\le d^{1-\frac2p} + \|s\|_p^{p-2} = d^{1-\frac2p} + \Phi(s)^{1-\frac2p} \\
&\le 2\Phi(s)^{1-\frac2p},
\end{align*}
where the final inequality follows from the fact that $\Phi(s) \ge \|w\|_1 \ge d.$
\end{proof}

\begin{lemma}[Progress step]
\label{lemma:progress}
Let $s^\new = s + \alpha|\mA z|$, as defined in line \ref{line:progresss} of \cref{algo:pnormbasic}.
Then we have that $\E(s^\new) \ge \E(s)$ and
\begin{align} \Phi(s^\new) - \Phi(s) \le 5p\alpha\Phi(s)^{1-\frac1p} + 3p^p\alpha^p\tau. \label{eq:progressbound} \end{align}
\end{lemma}
\begin{proof}
To bound $\E(s^\new)$, note that  $s^\new \geq s \geq \vec{0}$ entrywise. Therefore, $\E(s^\new) \ge \E(s)$.

To bound $\Phi(s^\new)$ we compute
\begin{align*}
\Phi(s^\new) - \Phi(s) &= \|s+\alpha|\mA z|\|_p^p - \|s\|_p^p \overset{(i)}{\le} 3p\alpha\sum_{i \in [n]} s_i^{p-1}|(\mA z)_i| + 3p^p\alpha^p\|\mA z\|_p^p \\
&\overset{(ii)}{\le} 3p\alpha\left(\sum_{i \in [n]} s_i^p\right)^{1/2}\left(\sum_{i \in [n]} s_i^{p-2}(\mA z)_i^2 \right)^{1/2} + 3p^p \alpha^p\tau \\
&\overset{(iii)}{\le} 3p\alpha\sqrt{\Phi(s)\E(s)} + 3p^p\alpha^p\tau \overset{(iv)}{\le} 5p\alpha\Phi(s)^{1-\frac1p} + 3p^p\alpha^p\tau.
\end{align*}
Here, $(i)$ follows from \cref{lemma:cor1} for $k = p$, $(ii)$ follows from the Cauchy-Schwarz inequality, $(iii)$ follows from the fact that $z$ is the minimizer for $\E(s)$, and $(iv)$ follows from \cref{lemma:energyupper} that $\Phi(s) \le 2\E(s)^{1-\frac2p}.$
\end{proof}

To analyze the boosting step we provide a general lemma about energy increase under boosting edges. Interestingly, this allows for resistances to increase by more than a constant factor, thereby going beyond the standard energy boosting lemmas in \cite{CKMST11,Mad13,Mad16}.

\begin{lemma}[Energy increase]
\label{lemma:energyincrease}
Let $w \in \R^n_{\geq 0}$ be $\ell_p$ Lewis weight overestimates for $\mA \in \R^{n \times d}$ and $\mD \se \mW^{1-\frac2p}$ be a diagonal matrix, and $v \in \R^n_{\ge0}$ satisfy $\|v\|_{\frac{p}{p-2}} \le 1$. For $\E \defeq \min_{g^{\top}x=-1} x^{\top}\mA^{\top}\mD\mA x$, $\E^\new \defeq \min_{g^{\top}x=-1} x^{\top}\mA^{\top}(\mD+\mV)\mA x$, and $y \defeq \argmin_{g^{\top}x=-1} x^{\top}\mA^{\top}\mD \mA x$ the following holds
\[ \E^\new - \E \ge \frac12 \sum_{i \in [n]} v_i(\mA y)_i^2. \]
\end{lemma}
\begin{proof}
By \cref{lemma:energysolve}
\[
\E = (g^{\top}(\mA^{\top}\mD \mA)^{-1}g)^{-1}
\text{ and } 
y = -\frac{1}{g^{\top}(\mA^{\top}\mD\mA )^{-1}g}(\mA^{\top}\mD \mA)^{-1}g.
\] 
By the Woodbury matrix identity we have that
\begin{align}
\E^{-1}-(\E^\new)^{-1} &= g^{\top}(\mA^{\top}\mD\mA )^{-1}\mA^{\top}\mV^\frac12(\mI+\mV^\frac12\mA(\mA^{\top}\mD\mA )^{-1}\mA^{\top}\mV^\frac12)^{-1}\mV^\frac12\mA(\mA^{\top}\mD\mA )^{-1}g \nonumber \\
&= \frac{1}{\E^2} y^\top \mA^\top \mV^\frac12(\mI+\mV^\frac12\mA(\mA^{\top}\mD\mA )^{-1}\mA^{\top}\mV^\frac12)^{-1}\mV^\frac12\mA y \label{eq:woodbury1}
\end{align}
We next claim that $\mA^\top\mV\mA \pe \mA^\top\mD\mA$. To show this, note that for any $x \in \R^n$ we have that
\begin{align*}
x^\top\mA^\top\mV\mA x = \sum_{i \in [n]} v_i(\mA x)_i^2 \overset{(i)}{\le} \|v\|_{\frac{p}{p-2}}\|\mA x\|_p^2 \overset{(ii)}{\le} x^\top\mA^\top\mW^{1-\frac2p}\mA x \overset{(iii)}{\le} x^\top \mA^\top \mD \mA,
\end{align*}
where $(i)$ follows from \holder's inequality, and $(ii)$ from the condition $\|v\|_{\frac{p}{p-2}} \le 1$ and \cref{lemma:weightkey}, and $(iii)$ from $\mW^{1-\frac2p} \pe \mD$. Note that this additionally implies that
\[ \mV^\frac12\mA(\mA^{\top}\mD\mA )^{-1}\mA^{\top}\mV^\frac12 \pe \mV^\frac12\mA(\mA^{\top}\mV\mA )^{-1}\mA^{\top}\mV^\frac12 \pe \mI \]
where the last step follows because the matrix is an orthogonal projection matrix.

Applying these bounds to \eqref{eq:woodbury1} yields that
\[ \E^{-1}-(\E^\new)^{-1} \ge \frac{1}{2\E^2} \sum_{i \in [n]} v_i(\mA y)_i^2. \]
Using that $\E^\new \ge \E$ and rearranging yields that
\[ \E^\new - \E \ge \frac{\E^\new\E}{2\E^2} \sum_{i \in [n]} v_i(\mA y)_i^2 \ge \frac12 \sum_{i \in [n]} v_i(\mA y)_i^2. \]
\end{proof}

\begin{lemma}[Boosting step]
\label{lemma:boosting}
Let $s$ be at the start of a boosting step, and $s^\new$ be defined as after the operations of line \ref{line:boost} in \Oracle~(\cref{algo:pnormbasic}). If $2^p\kappa^{-(p-2)}\Phi(s)^{1-\frac2p} \le \tau/4$ then $\Phi(s^\new) - \Phi(s) \le 20\kappa^2(\E(s^\new) - \E(s))$ and $\E(s^\new) - \E(s) \ge \tau^{2/p}/16.$
\end{lemma}
\begin{proof}
For $z$ as in line \ref{line:setz} of \cref{algo:oracle}
\begin{align} \sum_{i \in S} |(\mA z)_i|^p &= \|\mA z\|_p^p - \sum_{i \notin S} |(\mA z)_i|^p \overset{(i)}{\ge} \|\mA z\|_p^p - \left(2^{-\frac{p}{p-2}}\kappa\right)^{-(p-2)}\sum_{i \notin S} s_i^{p-2}(\mA z)_i^2 \\
&\overset{(ii)}{\ge} \|\mA z\|_p^p - 2^{p+1}\kappa^{-(p-2)}\Phi(s)^{1-\frac2p} \overset{(iii)}{\ge} \|\mA z\|_p^p/2, \label{eq:boostingpnormlower} \end{align}
where $(i)$ follows by the definition of $S$ in line \ref{line:chooseS} in \Oracle~(\cref{algo:pnormbasic}), $(ii)$ follows from \cref{lemma:energyupper}, and $(iii)$ follows by the condition on $\kappa$ in the hypothesis and $\tau \le \|\mA z\|_p^p$ by the condition of line \ref{line:iftau} in \Oracle~(\cref{algo:pnormbasic}).

Now we can lower bound $\E(s^\new) - \E(s)$ using \cref{lemma:energyincrease}. Set $\mD = d^{1-\frac2p}\mR + \mS^{p-2}$ and $v_i = 0$ for $i \notin S$ and $v_i = \frac{\tau^{2/p}|(\mA z)_i|^{p-2}}{4\|\mA z\|_p^p}$ for $i \in S$. Note that $\|v\|_{\frac{p}{p-2}} \le \tau^{2/p}/\|\mA z\|_p^2 \le 1$ by the condition $\tau \le \|\mA z\|_p^p$ of line \ref{line:iftau} in \Oracle~(\cref{algo:pnormbasic}). Thus \cref{lemma:energyincrease} gives
\begin{align*}
\E(s^\new) - \E(s) \ge \frac12 \sum_{i\in[n]} v_i(\mA z)_i^2 = \frac12 \cdot \frac{\tau^{2/p}}{4\|\mA z\|_p^p}\sum_{i \in S} |(\mA z)_i|^p \ge \tau^{2/p}/16
\end{align*}
where we have used \eqref{eq:boostingpnormlower} above.

To bound $\Phi(s^\new) - \Phi(s)$, we use \cref{lemma:cor2} for $k = p/(p-2)$ and $a = s_i^{p-2}$ and $b = v_i$ to get
\begin{align*}
\Phi(s^\new) - \Phi(s) &= \sum_{i\in[n]} \left((s_i^{p-2} + v_i)^\frac{p}{p-2} - s_i^p\right) \le 4^\frac{p}{p-2}\sum_{i \in S} \left(s_i^2v_i + v_i^\frac{p}{p-2}\right) \\
&\overset{(i)}{\le} 4^\frac{p}{p-2}\sum_{i \in S}\left(4^{-\frac{p}{p-2}}\kappa^2 v_i(\mA z)_i^2 + v_i^\frac{p}{p-2}\right) \overset{(ii)}{\le} 2\kappa^2(\E(s^\new) - \E(s)) + 1.
\end{align*}
Here, $(i)$ uses $s_i \le 2^{-\frac{p}{p-2}}\kappa|(\mA z)_i|$ for all $i \in S$ by line \ref{line:chooseS} in \Oracle~(\cref{algo:pnormbasic}) and $(ii)$ uses \cref{lemma:energyincrease} and $\|v\|_{\frac{p}{p-2}} \le 1/4$. To conclude, note that $\E(s^\new) - \E(s) \ge \tau^{2/p}/16 \ge 1/16$, as $\tau \ge 1$. Also, $\kappa \ge 1$, so $2\kappa^2(\E(s^\new) - \E(s)) + 1 \le 20\kappa^2(\E(s^\new) - \E(s))$. This completes the proof.
\end{proof}

Now we can combine the bounds on $\Phi(s)$ and $\E(s)$ in \cref{lemma:progress,lemma:boosting} to prove \cref{lemma:largemain}.
\begin{proof}[Proof of \cref{lemma:largemain}]
We set $\tau_p = 40^p$. Choose $\alpha_p = 1/(1000p)$ so that $p^p\alpha^p\tau \le p\alpha d^{1-\frac1p}$. Then $p^p\alpha^p\tau \le p\alpha \Phi(s)^{1-\frac1p}$ as $\Phi(s) \ge d$ for all $s \ge w^{1/p}$ for a Lewis weight overestimate $w$. Let $\kappa_p = p.$

Let $s^\final$ be the final value of $s$ in a call to \Oracle~(\cref{algo:pnormbasic}). We show by induction that $\Phi(s^\final) \le (20\kappa)^p$ and that $2^p\kappa^{-(p-2)}\Phi(s)^{1-\frac2p} \le \tau/4$ during a successful execution of \cref{algo:oracle} always (so the condition of \cref{lemma:boosting} is satisfied).
We start by bounding $\Phi(s^\final)$. As $\sum_{i \in [n]} w_i \le 2d$ by the definition of $\ell_p$ Lewis weight overestimates (\cref{def:overestimate}), initially $\Phi(s) \le 2d$. We calculate that
\begin{align*}
\Phi(s^\final) &\overset{(i)}{\le} 2d + \alpha^{-1}d^{1/p}\left(5p\alpha\Phi(s^\final)^{1-\frac1p} + 3p^p\alpha^p\tau\right) + 20\kappa^2\E(s^\final) \\
&\overset{(ii)}{\le} 2d + \alpha^{-1}d^{1/p} \cdot 8p\alpha\Phi(s^\final)^{1-\frac1p} + 40\kappa^2\Phi(s^\final)^{1-\frac2p} \\
&= 2d + 8p\Phi(s^\final)^{1-\frac1p}d^{1/p} + 40\kappa^2\Phi(s^\final)^{1-\frac2p}.
\end{align*}
where $(i)$ follows from \cref{lemma:progress,lemma:boosting}, and $(ii)$ follows from \cref{lemma:energyupper} and the bound $p^p\alpha^p\tau \le p\alpha \Phi(s)^{1-\frac1p}$ from our choice of $\tau_p$ and $\alpha_p$.
If $\Phi(s^\final) > (20\kappa)^p$ then we get that
\[ 2d\Phi(s^\final)^{-1} + 8pd^{1/p}\Phi(s^\final)^{-\frac1p} + 40\kappa^2\Phi(s^\final)^{-\frac2p} < \frac{1}{20} + \frac{8pd^{1/p}}{20\kappa} + \frac{40\kappa^2}{400\kappa^2} < 1,\]
contradicting the above equation. Hence $\Phi(s^\final) \le (20\kappa)^p$.

Now we check that $2^p\kappa^{-(p-2)}\Phi(s^\final)^{1-\frac2p} \le \tau/4$ to complete the induction. From the choice $\tau_p = 40^p$ and $\tau \ge \tau_p$, note that
\[ 2^p\kappa^{-(p-2)}\Phi(s^\final)^{1-\frac2p} \le 2^p\kappa^{-(p-2)}(20\kappa)^{p-2} \le 40^p / 4 \le \tau/4. \]

We now show that the returned vector $x = (\alpha T)^{-1}y$ (for $T \defeq \lfloor \alpha^{-1}d^{1/p}\rfloor$) satisfies $\|\mA x\|_p \le O(p)$ and $x^\top\mA^\top\mR\mA x \le O(p)^p$. Note that $(\alpha T)^{-1} \le 2d^{-1/p}.$ For the first of these note that 
\[ \|\mA x\|_p \le 2d^{-1/p}\|\mA s^\final\|_p = 2d^{-1/p}\Phi(s^\final)^{1/p} \le 40\kappa d^{-1/p} \le 40p \] by the choice of $\kappa_p$.
For the latter, note that \[ z^\top\mA^\top\mR\mA z \le d^{\frac2p-1}\E(s) \le 2d^{\frac2p-1}\Phi(s)^{1-\frac2p} = O(p)^p \] at each step -- now apply the triangle inequality on the norm $\|\mR^{1/2}\mA z\|_2.$

Finally we bound the number of progress and boosting steps. The number of progress steps is bounded by $\alpha^{-1}d^{1/p} = O\left(pd^\frac{p-2}{3p-2}\right)$ by the choice of $\alpha$. To bound the number of boosting steps, note that $\E(s)$ increases by $\tau^{2/p}/16$ per boosting step by \cref{lemma:boosting}, and is increasing every progress step by \cref{lemma:progress}. As $\E(s^\final) \le 2\Phi(s)^{1-\frac2p} \le 2(20\kappa)^{p-2}$ at the end we get that the number of boosting steps is bounded by
\[ \frac{2(20\kappa)^{p-2}}{\tau^{2/p}/16} \le O(p)^p \cdot d^{1-\frac2p} \cdot d^{\frac{-2(p-2)(p-1)}{p(3p-2)}} = O(p)^p d^\frac{p-2}{3p-2}. \]
To compute the $\ell_p$ Lewis weights overestimates in line \ref{line:approxlewislarge} in \Oracle~(\cref{algo:pnormbasic}) there are an additional $\O(1)$ solves to $\mA^\top\mD\mA$ by \cref{lemma:approxplarge}. Together, this gives the total iteration bound.
\end{proof}

%% file: largep_ms.tex
\section{Monteiro-Svaiter Acceleration Algorithm for Large $p$}
\newcommand{\xhat}{\hat{x}}
\newcommand{\prox}{\mathsf{Prox}}
\newcommand{\cprox}{\mathsf{CProx}}
\newcommand{\pgd}{\mathsf{ProxOracle}}
\newcommand{\mh}{\mathbf{H}}
\newcommand{\magic}{\mathsf{Magic-}\lambda}
\label{sec:largep_ms}

In \cref{sec:largep_ckmst}, we gave an algorithm for $\ell_p$ regression for $p \geq 2$  based on the iterative refinement framework of \cite{AKPS19}. Here we give an alternate scheme with an improved dependence on $p$ based on \emph{highly-smooth} optimization. More specifically, we leverage an optimization framework from \cite{BJLLS19}, which reduces the task of minimizing a convex function $f$ to approximately solving proximal subproblems of the form 
\[
\prox(y) = \min_{x} f(x) + C_p \norm{x-y}_\mm^p
\]

for arbitrary positive semidefinite matrix $\mm$. Our result is a refinement of the $O(p^{14/3} n^{1/3})$ iteration complexity achieved in \cite{CJJJLST20}. Our main technical ingredient is an improved Hessian stability bound (\Cref{lemma:hessian_stable}) which works for all $p \geq 2$ and allows us to take steps bounded in the norm induced by a matrix $\mm \preceq \mA^\top \mA$. We leverage this to give an efficient algorithm for proximal subproblems, and combine with the acceleration framework of \cite{BJLLS19} to obtain our result.

\subsection{Hessian stability}
In this section, we prove our Hessian stability bound \Cref{lemma:hessian_stable}. We begin with a straightforward scalar inequality which we use in our proof. 

\begin{lemma}
\label{lemma:scalar_ineq}
Let $\alpha, \beta \geq 1$ satisfy $\frac{1}{\alpha} + \frac{1}{\beta} = 1$. For any $n \geq 0$ and any $x,y$, 
\[
\left| x + y \right|^n \leq |\alpha x|^n + |\beta y|^n.
\]
Additionally, $\left| x+y \right|^{p-2} \leq e \left|x\right|^{p-2} + p^{p-2} \left| y \right|^{p-2}$ for $p \geq 2$.
\end{lemma}
\begin{proof}
Observe
\[
\left| x+ y \right|^n = \left| \frac{\alpha x}{\alpha} + \frac{\beta y}{\beta} \right|^n \leq \left| \max \left\{ |\alpha x|, |\beta y| \right\} \right|^n = \max \left\{ |\alpha x|^n, |\beta y|^n \right\} \leq |\alpha x|^n + |\beta y|^n.
\]
Applying this inequality with $\alpha = \frac{p-1}{p-2}$, $\beta = p-1$, and $n = p-2$ yields
\[
\left| x+y \right|^{p-2} \leq \left(1 + \frac{1}{p-2}\right)^{p-2} \left|x\right|^{p-2} + \left| (p-1) y \right|^{p-2} \leq e  \left|x\right|^{p-2} + p^{p-2} \left| y \right|^{p-2}
\]
where the last inequality follows from $\left(1+\frac{1}{x}\right)^x < e$ for any $x \geq 0$ and $p-1 \leq p$. 
\end{proof}

With this scalar inequality, we define a matrix $\mm$ we will repeatedly appeal to in this section.

\begin{definition}
\label{def:magic-m}
Let $\ma \in \R^{n \times d}$ be a matrix, and let $w \in \R^n$ be a vector of overestimates of the $\ell_p$-Lewis weights of $\ma$ (\cref{def:overestimate}). We set $\mm \defeq \ma^\top \mw^{1-2/p} \ma.$
\end{definition}

With this, we prove our main Hessian stability fact \cref{lemma:hessian_stable}:

\begin{lemma}
\label{lemma:hessian_stable}
Let $p \geq 2$, and define $f(x) = \norm{\ma x- b}_p^p$. Let $\mM = \mA^\top \mW^{1-2/p}\mA$ (\Cref{def:magic-m}). For any $y \in \R^d$, define $f_{y}(x) = f(x) + C_p \norm{x-y}_\mm^p$ and $h_y(x) = \norm{x-y}_{\nabla^2 f(y)}^2 + C_p \norm{x-y}_{\mm}^p.$ Then if $C_p = e \cdot p^p$, for any $x$
\[
\frac{1}{e} \nabla^2 h_y(x) \preceq \nabla^2 f_y(x) \preceq  e \nabla^2 h_y(x). 
\]
\end{lemma}
\begin{proof}
We first note
\[
\hess f(x)= p (p-1) \ma^\top \diag{\left| \ma x - b\right|}^{p-2} \ma. 
\]
For any vector $z$, we use \Cref{lemma:scalar_ineq} to get
\begin{align*}
z^\top \hess f(x) z &= p (p-1) \sum_{i\in[n]}  |\ma x - b|_i^{p-2} \left(\ma z\right)_i^2 \\
&= p (p-1) \sum_{i\in[n]} |\ma y - b + \ma(x-y)|_i^{p-2} \left(\ma z\right)_i^2 \\
&\leq \sum_{i\in[n]}  \left( ep(p-1) |\ma y - b|_i^{p-2}  + p^p |\ma(x-y)|_i^{p-2} \right) \left(\ma z\right)_i^2.
\end{align*}
Now, by H\"older's inequality and \Cref{lemma:weightkey} we get
\begin{align*}
\sum_{i\in[n]}  p^p \left|\ma (x-y) \right|_i^{p-2} (\ma z)_i^2 &\leq p^p \norm{ \left|\ma (x-y) \right|^{p-2}}_{\frac{p}{p-2}} \norm{(\ma z)^2}_{\frac{p}{2}} \\
&=  p^p \norm{\ma (x-y)}_p^{p-2} \norm{\ma z}_{p}^2 \\
&\leq p^p \norm{x-y}_\mm^{p-2} \norm{z}_\mm^2. 
\end{align*}
Combining the above two inequalities yields 
\begin{align}
z^\top \hess f(x) z &\leq e p(p-1) z^\top \ma^\top \diag{|\ma y - b|^{p-2}} \ma z + p^p \norm{x-y}_\mm^{p-2} \norm{z}_\mm^2 \nonumber \\
&= e \norm{z}_{\nabla^2 f(y)}^2 + p^p \norm{x-y}_\mm^{p-2} \norm{z}_\mm^2. \label{eqn:ineq}
\end{align}
Define $g_y(x) = C_p \norm{x-y}_\mm^p$. 
We have
\begin{align*}
\nabla^2 g_y(x) &= p C_p \norm{x-y}_{\mm}^{p-2} \mm + p(p-2) C_p \norm{x-y}_{\mm}^{p-4} \mm (x-y) (x-y)^\top \mm
\end{align*}
and thus
\[
p C_p \norm{x-y}_\mm^{p-2} \mm \preceq \nabla^2 g_y(x).
\]
Combining the two inequalities yields
\[
\nabla^2 f_y(x) = \nabla^2 f(x) + \nabla^2 g_y(x)
\preceq e (\nabla^2 h_y(x) - \nabla^2 g_y(x)) + \frac{1}{ep}  \nabla^2 g_y(x) + \nabla^2 g_y(x) 
\preceq e \nabla^2 h_y(x).
\]
For the lower bound, we exchange $x$ and $y$ in \Cref{eqn:ineq} and obtain
\[
z^\top \nabla^2 f(x) z \geq \frac{1}{e} z^\top \nabla^2 f(y) z - \frac{p^p}{e} \norm{x-y}_\mm^{p-2} \norm{z}_\mm^2. 
\]
Consequently,
\[
	\nabla^2 f_y(x) =	 \nabla^2 f(x) + \nabla^2 g_y(x)
	 \succeq \frac{1}{e} (\nabla^2 h_y(x) - \nabla^2 g_y(x)) -  \frac{1}{ep}  \nabla^2 g_y(x) 
 	+ \nabla^2 g_y(x)
	\succeq \frac{1}{e} \nabla^2 h_y(x).
\]
\end{proof}
\subsection{Efficient implementation of proximal subproblems}

We now leverage \Cref{lemma:hessian_stable} to give an efficient oracle for the problem 
\[
\prox(y) = \argmin_{x} \norm{\ma x- b}_p^p + e p^p \norm{x-y}_\mm^p. 
\]
Our algorithm is based on the \emph{relative smoothness} framework from \cite{relsmooth}. We use the following:

\begin{lemma}[Theorem~3.1 from \cite{relsmooth}]
Let $f,h$ be convex twice-differentiable functions satisfying
\[
\mu \nabla^2 h(x) \preceq \nabla^2 f(x) \preceq L \nabla^2 h(x)
\]
for all $x$. There is an algorithm which given a point $x_0$ computes a point $x$ with
\[
f(x) - \argmin_y f(y) \leq \eps \left( f(x_0) - \argmin_y f(y) \right) 
\]
in $O(\frac{L}{\mu} \log(1/\eps))$  iterations, where each iteration requires computing gradients of $f$ and $h$ at a point, $O(n)$ additional work, and solving a subproblem of the form 
\begin{equation}
\label{eqn:prox_prob}
\min \left\{ \left\langle  g, x \right\rangle + L h(x) \right\}
\end{equation}
for vectors $g$.
\label{lem:relsmooth}
\end{lemma}
Applying this to the $p$-norm regression objective yields the following result.
\begin{lemma}
Let $\ma \in \R^{n \times d}, b \in \R^n$ be given. Let $\mm = \ma^\top \mw^{1-2/p} \ma$ (\Cref{def:magic-m}). There exists an algorithm which 
computes $\argmin_x \norm{\ma x-b}_p^p + e p^p \norm{x-y}_\mm^p$ to high accuracy using $\otilde(1)$ linear system solves on matrices $\ma^\top \md \ma$ for $\md \succeq 0$, i.e. in $\O(\T_A)$ time.
\label{lemma:prox}
\end{lemma}
\begin{proof}
For the function $f_y(x) = \norm{\ma x - b}_p^p + e p^p \norm{x-y}_\mm^p$, we define the regularizer $h_y(x) = \norm{x-y}_{\nabla^2 f(y)}^2 + e p^p \norm{x-y}_\mm^p$. We observe by \Cref{lemma:hessian_stable} that $\nabla^2 f_y(x) \approx_{O(1)} \nabla^2 h_y(x)$ for all $x$. Thus \Cref{lem:relsmooth} ensures we compute a minimizer to $f_y$ using $\otilde(1)$ calls to an oracle which solves subproblems of the form
\[
\min \left\{ \left\langle  g, x-z \right\rangle + 4 \left(  \norm{x-y}_{\nabla^2 f(y)}^2 + e p^p \norm{x-y}_\mm^p \right)   \right\}.
\]
To solve this problem, we employ the algebra fact that $\frac{1}{s}x^s = \max_{y\geq 0} xy- \frac{1}{r} y^r$ for any $x \geq 0$ and $\frac{1}{s} + \frac{1}{r} = 1$. Thus, we have 
\[
\norm{x-y}_\mm^p = \frac{p}{2} \cdot \frac{2}{p} \left( \norm{x-y}_\mm^2 \right)^{p/2} = \frac{p}{2} \max_{\tau\geq 0} \left\{ \tau \norm{x-y}_\mm^2 - \frac{p-2}{p} \tau^{\frac{p}{p-2}}\right\}.
\]
We may therefore write the subproblem as
\[
\min_x \max_{\tau \geq 0} \left\{ \left\langle  g, x-z \right\rangle + 4 \norm{x-y}_{\nabla^2 f(y)}^2 + 2ep^{p+1} \left( \tau \norm{x-y}_\mm^2 - \frac{p-2}{p} \tau^{\frac{p}{p-2}} \right) \right\}.
\]
This problem is convex in $x$ and concave in $\tau$: we may exchange the min and max above. Further, this is a convex quadratic in $x$, and thus for any fixed $\tau$ we may compute the minimizing $x$ with a single linear system solve of the form $\nabla^2 f(y) + C \mm =  \ma^\top \md \ma$, for some constant $C \geq 0 $ and $\md \succeq 0$. Further, for any $C > 0$ we have $\nabla^2 f(y) + C \mm \succ 0$: for any fixed $\tau > 0$ the minimizing value of $x$ is unique. We conclude by binary searching for $\tau$ to high accuracy. Thus, each proximal subproblem may be solved using $\otilde(\T_\mA)$ time.
\end{proof}
We note that a high-accuracy solution to the proximal problem in \cref{lemma:prox} gives an approximate stationary point (exactly the condition later in \cref{def:prox}).

\subsection{Putting it all together}

We finish by using the above subroutine in the acceleration framework of \cite{BJLLS19}. We summarize the main claim here:

\begin{definition}[Approximate Proximal Step Oracle, Definition~$5$ \cite{BJLLS19}]
\label{def:prox}
We call $\mathcal{O}_{prox}$ an \emph{$(\alpha,\delta)$-approximate proximal oracle} for convex $f : \R^d \rightarrow \R$ if, when queried at any $x \in \R^d$ it returns $y = \mathcal{O}_{prox}(x) \in \R^d$ such that \[ 
\norm{\nabla f(y)  + e p^{p+1} \norm{y-x}_\mm^{p-2}\cdot \mm ( y-x) }_{\mm^{-1}} \leq e \alpha p^{p+1} \norm{x-y}_\mm^{p-1}  + \delta
\]
\end{definition}

\begin{theorem}[Theorem 7 from \cite{BJLLS19}]
Let $g : \R^d \rightarrow \R $ be a convex twice-differentiable function minimized at $\xopt$, and let $x_0$ be a point with $\norm{x_0 - \xopt} \leq R$. For any parameter $\eps \geq 0$, there is an algorithm which for all $k$ computes $x$ with 
\[
f(x) - f(\xopt) \leq \max \left\{ \eps, \frac{100 p^p \cdot 40^{p-2} R^p }{k^{\frac{3p-2}{2}}} \right\}
\]
using $\lceil k ( 6 + \log_2 [ 10^{20}  R^p \cdot (10^5 p)^{p+6} \eps^{-1} ])^2\rceil = O(p^2 k \log^2(p R \eps^{-1/p} ))$ gradients of $f$ and queries to an $(\frac{1}{128 p^2} , \delta)$-approximate proximal oracle, provided that both $\delta \leq \eps /[10^{20} p^2 R]$ and $\eps \leq 10^{20} p^p \gamma^4 R^{p+1}$. 
\label{thm:ms}
\end{theorem}
\begin{proof}
Define the convex function $g(x) =  f(\mm^{-1/2} x)$, and choose $\omega(x) = e p^p x^{p-2}$. The optimality conditions of \Cref{def:prox} are equivalent to those in Definition~5 of \cite{BJLLS19} after applying this change of basis. \Cref{thm:ms} then follows from applying Theorem~7 in \cite{BJLLS19} to $g$ with $\gamma = p$ and $\alpha = \frac{1}{128 p^2}$. 
\end{proof}

Our application of this fact relies on a diameter-shrinking argument from \cite{CJJJLST20}. We first recall a standar bound on the strong convexity of $\norm{x}_p^p$, which we cite from \cite{AKPS19} for simplicity.

\begin{lemma}[Lemma 4.5 from \cite{AKPS19}]
\label{lemma:adil}
Let $p \in (1,\infty)$. Then for any two vectors $y, \Delta \in \R^n$, 
\[ 
\norm{y}_p^p + v^\top \Delta + \frac{p-1}{p 2^p} \norm{\Delta}_p^p \leq \norm{y + \Delta}_p^p 
\]
where $v_i = p |y|_i^{p-2} y_i$ is the gradient of $\norm{y}_p^p$. 
\end{lemma}

We finally need the following lemma which allows us to convert points which low function error into points with small distance to the minimizer.
\begin{corollary}
\label{corr:fn_to_dist}
Let $\ma \in \R^{n \times d}, b \in \R^n$ be given. Let $\mm = \ma^\top \mw^{1-2/p} \ma$ for $\ell_p$ Lewis weight overestimates $w$ (\Cref{def:lewisweights}). Let $f(x) = \norm{\ma x-b}_p^p$ be minimized at $\xopt$. If $x$ satisfies $f(x) - f(\xopt) \leq \Eps$, then $\norm{x - \xopt}_\mM \leq 2^{3/2} d^{1/2 - 1/p} \Eps^{1/p}$. 
\end{corollary}
\begin{proof}
Applying \Cref{lemma:adil}, we have 
\[
\norm{\ma x - b}_p^p = \norm{\ma \xopt - b + \ma \left(x -\xopt\right)}_p^p \geq \norm{\ma \xopt -b}_p^p + \nabla f(\xopt)^\top (x - \xopt) + \frac{p-1}{p 2^p} \norm{\ma (x - \xopt)}_p^p.
\]
Note $\nabla f(\xopt) = 0$ by optimality of $\xopt$. By \Cref{eq:upperp}, we obtain
\[
\norm{x - \xopt}_\mM^2 \leq \norm{w}_1^{1 - 2/p} \norm{\ma (x-\xopt)}_{p}^2.
\]
by H\"older's inequality and the fact that $p \geq 2$. Recall that $\norm{w}_1 \leq 2 d$ by \Cref{lemma:approxplarge}: this implies
\[
\norm{x - \xopt}_\mM^p \leq (2d)^{p/2-1} \norm{\ma (x-\xopt)}_{p}^p.
\]
Thus,
\[
\Eps \geq f(x) - f(\xopt) \geq \frac{p-1}{p 2^p} \norm{\ma (x-\xopt)}_p^p \geq 2^{-p-1} (2d)^{1-p/2} \norm{x - \xopt}_\mM^p:
\]
taking $p^{th}$ roots yields $\norm{x - \xopt}_\mm\leq 2^{3/2} d^{1/2 - 1/p} \Eps^{1/p}$ as desired.
\end{proof}

We now prove the main decrease lemma which in turn shows \cref{thm:largemain}.
\begin{lemma}
\label{lemma:reduction}
Let $\ma, b$ be given, and let $f(x) = \norm{\ma x - b}_p^p$ have minimizer $\xopt$. Let $x_0$ be a point such that $f(x_0) - f(\xopt) \leq \Eps$. There is an algorithm which returns $x'$ with $f(x') - f(\xopt) \leq \frac{\Eps}{2}$ using
\[
\widetilde{O} \left(p^{8/3} d^{\frac{p-2}{3p-2}} \T_\mA \right)
\]
time.
\end{lemma}
\begin{proof}
Applying \Cref{corr:fn_to_dist} yields  $R \equiv \norm{x_0 - \xopt}_\mm \leq 2^{3/2} d^{1/2 - 1/p} \Eps^{1/p}$ for $\mM = \mA^\top \mW^{1-2/p}\mA$. Note that $\log(p R \Eps^{-1/p}) = O(\log d)$. We now apply \Cref{thm:ms} to $f(x)$ with $\eps = \frac{\Eps}{2}$:  in $\otilde(p^2 k)$ gradient computations and proximal oracle calls we compute $x$ with 
\[
f(x) - f(\xopt) \leq \max \left\{ \frac 12 \Eps, \frac{100 p^p \cdot 40^{p-2} R^p }{k^{\frac{3p-2}{2}}} \right\} \leq  \max \left\{\frac 12 \Eps, \frac{100 (120 p)^p d^{p/2 - 1} \Eps }{k^{\frac{3p-2}{2}}} \right\}.
\]
For $k = O(p^{2/3} d^{\frac{p-2}{3p-2}})$, this bound is $\frac 12 \Eps$ as desired. We additionally require $\otilde(p^2 k ) = \otilde(p^{8/3} d^{\frac{p-2}{3p-2}})$ gradient computations and calls to a proximal oracle for $f$ -- these proximal oracle calls can each be implemented in $\widetilde{O}(\T_\mA)$ time by \Cref{lemma:prox}.
\end{proof}

\begin{proof}[Proof of \cref{thm:largemain}]
Let $\xopt = \argmin_y \|\mA y - b\|_p^p$ and $\OPT = \|\mA \xopt - b\|_p^p.$ We may initialize $\Eps = n^\frac{p-2}{2} \OPT$ in \cref{lemma:reduction} by setting  $x = \argmin_x \|\mA y - b\|_2^p$ instead, and noting that \[ \|\mA y - b\|_p^p \le \|\mA y - b\|_2^p \le \|\mA \xopt - b\|_2^p \le n^\frac{p-2}{2} \|\mA \xopt - b\|_p^p = n^\frac{p-2}{2} \OPT. \]
Now \cref{thm:largemain} follows from running $\log(n^\frac{p-2}{2}) = \O(p)$ iterations of \cref{lemma:reduction}.
\end{proof}

\paragraph{Discussion on numerical stability.} Throughout the section (eg. in the application of \cref{lem:relsmooth,lemma:prox}), we have assumed that high accuracy solutions to problems lead to exact or high accuracy stationary points, i.e. the KKT conditions are satisfied. There are several ways to make this rigorous. In particular, if one assumes that all parameters, including the condition number of $\mA$, are quasipolynomially bounded (i.e. at most $\exp(\poly\log m)$), then one can add a small strongly-convex regularizer (eg. $\delta\|x\|_{\mA}^2$ for $\delta \le \exp(-\poly\log m)\eps$) which barely affects the optimal value. Strong convexity allows us to get an approximate stationary point from approximate minimizers, which suffices for the all our applications (including the proof of \cite{relsmooth}).

%% file: smallq.tex
\section{Algorithm for Small $q$}
\label{sec:smallq}

In this section, we provide an algorithm to show \cref{thm:smallmain}. Because there isn't a clean version of 
iterative refinement for the objective $\|\mA x - b\|_q$ for $q < 2$, we instead work with the dual problem. Precisely, we can use Sion's minimax theorem to get for $p = q/(q-1)$
\begin{align}
\min_{x \in \R^d} \|\mA x - b\|_q &= \min_{x \in \R^d} \max_{\|y\|_p \le 1} y^\top(\mA x - b) = \max_{\|y\|_p \le 1} \min_{x \in \R^d} y^\top(\mA x - b) \nonumber \\
&= -\min_{\substack{\|y\|_p \le 1 \\ \mA^\top y = 0}} b^\top y = \left(\min_{\mA^\top y = 0, b^\top y = 1} \|y\|_p\right)^{-1}. \label{eq:dualprob}
\end{align}
Using an high precision solution $y$ to \eqref{eq:dualprob}, we can return a high precision minimizer to $\min_{x \in \R^d} \|\mA x - b\|_q.$ In particular for the true optimum $y^*$, by KKT conditions (that $\g \|y^*\|_p^p = p\sign(y^*)|y^*|^{p-2}$ is in the kernel of $\begin{bmatrix} \mA & b \end{bmatrix}^\top$) we know that there exists a vector $x \in \R^d$ satisfying $\lambda \sign(y^*)|y^*|^{p-2} = \mA x - b.$ We return this $x$. If we have a high precision minimizer $y$ instead of the true optimum $y^*$, we can instead return an $\ell_2$-projection, i.e. $x = \argmin_{x \in \R^d} \|\mA x - b - \lambda \sign(y)|y|^{p-2}\|_2$ for the proper scaling $\lambda$.

We use \cref{lemma:gammasolver} (for $\mA = \mI$ and $b = 0$ in the lemma statement) to solve the problem in \eqref{eq:dualprob}, where we assume $\nu = 1$ and $\OPT = 1$ by scaling. This leads to the following optimization problem for some $g \in \R^n$ and $\mR = \diag{r}$ for $r \in \R^n_{>0}$:
\[ \min_{\mA^\top x = 0, b^\top x = 1, g^\top x = -1} x^\top \mR x + \|x\|_p^p. \]
Let $\mU = \begin{bmatrix} \mA & b & g \end{bmatrix}$ and $v = \begin{bmatrix} \mathbf{0} & 1 & -1 \end{bmatrix}^\top$. Through these reductions and \cref{lemma:gammasolver} it suffices to show the following.
\begin{lemma}
\label{lemma:smallmain}
For matrix $\mU \in \R^{n \times d}$ and $v \in \R^n$, assume there is $x \in \R^n$ satisfying $\mU^\top x = v$, $x^\top \mR x \le 1$ and $\|x\|_p^p \le 1$. Then there is an algorithm that in time $\O(\T_\mU)$ outputs a $y \in \R^n$ satisfying $\mU^\top y = v$, $y^\top \mR y = O(1)$, and $\|y\|_p \le O(d^\frac{p-2}{2p-2}).$
\end{lemma}
We remark that we could also instead get a result which achieves a $y^\top \mR y = O_p(1)$ and $\|y\|_p^p = O_p(1)$ in time $\O_p(\T_\mU d^\frac{p-2}{2p-2})$ via a multiplicative weights style algorithm as done in \cref{sec:largep_ckmst} algorithm \Oracle~(\cref{algo:pnormbasic}). However since both runtimes would be the same (up to logarithmic factors), we choose to present this simpler single iteration algorithm. Combining this multiplicative weights style algorithm with energy boosting as in the analysis in \cref{sec:largep_ckmst} to achieve a $O_p(d^\frac{p-2}{3p-2})$ iteration bound remains an interesting open problem.

\begin{proof}[Proof of \cref{thm:smallmain}]
\cref{lemma:smallmain} satisfies the conditions of \cref{lemma:gammasolver} for $\gamma = O(d^\frac{p-2}{2p-2}).$ Each call to \cref{lemma:smallmain} requires $\O(1)$ calls to a solver for $\mU^\top \mD \mU$.
Also, a solver for $\mU^\top \mD \mU$ can be implemented using $O(1)$ calls to a solver for $\mA^\top \mD \mA$ for diagonal matrices $\mD$, and $O(1)$ solves on $O(1) \times O(1)$ matrices by computing the inverse via the Schur complement onto the $2 \times 2$ block of $\mU^\top \mD \mU$ corresponding to the $b, g$ vectors. Thus the total number of iterations is $\O(p^{3.5}\gamma \log(m/\eps)) = \O_p(d^\frac{p-2}{2p-2} \log(m/\eps))$ as desired.
\end{proof}
An algorithm to show a weaker version of \cref{lemma:smallmain} with the bound $\|y\|_p \le O\left(n^\frac{p-2}{2p-2}\right)$ was given in \cite[Lemma 3.3]{APS19}, by simply returning \[ y = \argmin_{\mU^\top x = v} x^\top\left(n^{1-\frac2p}\mR + \mI\right)x. \] Our approach to improve this dependence to $O(d^\frac{p-2}{2p-2})$ uses a version of $\ell_q$ Lewis weights to replace the identity matrix $\mI$ in the above. To handle the presence of the resistance term $\mR$ we require a regularized version of Lewis weights.
\begin{definition}[Regularized $\ell_q$ Lewis weights]
\label{def:reglewis}
For a matrix $\mA \in \R^{n \times d}$, $1 \le q \le 2$, and vector $c \in \R^n_{\ge0}$, the $c$\emph{-regularized} $\ell_q$ Lewis weights $w$ are defined as the solution to
\[ w_i = \sigma\left(\left(\mC + \mW\right)^{\frac12-\frac1q}\mA \right)_i \enspace \text{ for all } \enspace i \in [n]. \]
\end{definition}
We show that these weights can be computed approximately in $\O(1)$ iterations of a contractive map. Each iteration requires the computation of approximate leverage scores. The proof of the following lemma is in \cref{proofs:reglevapprox}.
\begin{lemma}
\label{lemma:reglevapprox}
Given a matrix $\mA \in \R^{n \times d}$, $1 \le q \le 2$, and vector $c \in \R^n_{\ge0}$, let $w$ be the $c$-regularized $\ell_q$ Lewis weights. There is an algorithm $\ApproxRegLewis(\mA, c, q)$ that whp. computes a vector $\hw \in \R^n_{\ge0}$ satisfying $\hw_i/w_i \in [0.9, 1.1]$ for all $i \in [n]$ in $\O(\T_\mA)$ time.
\end{lemma}
We can now give our algorithm to show \cref{lemma:smallmain}.
\begin{algorithm}[ht]
\caption{$\OracleSmall(\mU, v, r, q).$ Given matrix $\mU \in \R^{n \times d}, v \in \R^n, r \in \R^n, q \le 2$, such that there exists $x$ with $\mU^\top x = v$ and $x^\top\mR x \le 1$ and $\|x\|_p^p \le 1$, returns $y$ satisfying $\mU^\top y = v$, $y^\top\mR y = O(1)$ and $\|y\|_p = O(d^\frac{p-2}{2p-2}).$ \label{algo:oraclesmall}}
$\hw \assign \ApproxRegLewis(\mU, d r^\frac{p}{p-2}, q).$ \Comment{\cref{lemma:reglevapprox}} \label{line:approxlewissmall} \\
Return $y \assign \argmin_{\mU^\top x = v} x^\top\left(d^{1-\frac2p}\mR + \hmW^{1-\frac2p}\right)x.$ \label{line:returny}
\end{algorithm}
We can now show \cref{lemma:smallmain}.
\begin{proof}[Proof of \cref{lemma:smallmain}]
Let $x$ satisfy $\mU^\top x = v$ and $x^\top\mR x \le 1$ and $\|x\|_p \le 1$.
By the definition of $y$ in line \ref{line:returny} in \OracleSmall~(\cref{algo:oraclesmall}), we know that
\begin{align*} y^\top\left(d^{1-\frac2p}\mR + \hmW^{1-\frac2p}\right)y &\le x^\top\left(d^{1-\frac2p}\mR + \hmW^{1-\frac2p}\right)x \le d^{1-\frac2p} + \|w\|_1^{1-\frac2p}\|x\|_p^2 \le 3d^{1-\frac2p}
\end{align*}
where we have used \holder~and $\|\hw\|_1 \le 1.1d$ by \cref{lemma:reglevapprox}. In particular, this gives us $d^{1-\frac2p}y^\top\mR y \le 3d^{1-\frac2p}$ so $y^\top\mR y \le 3$. Now we bound $\|y\|_p^p$. Note that by the optimality conditions for $y$ (that the gradient of the objective of line \ref{line:returny} of \cref{algo:oraclesmall} is in the kernel of $\mU^\top$), there must exist a vector $z \in \R^d$ such that
\[ y = \left(d^{1-\frac2p}\mR + \hmW^{1-\frac2p}\right)^{-1}\mU z. \] Define $\hy = \left(d^{1-\frac2p}\mR + \hmW^{1-\frac2p}\right)^{-1/2}\mU z$, so that \[ \|\hy\|_2^2 = y^\top\left(d^{1-\frac2p}\mR + \hmW^{1-\frac2p}\right)y \le 3d^{1-\frac2p}. \]
Now we get that
\begin{align}
\|y\|_p^p &= \sum_{i \in [n]} \left(d^{1-\frac2p}r_i + \hw_i^{1-\frac2p}\right)^{-p}|(\mU z)_i|^p \nonumber \\
&= \sum_{i \in [n]} \left(d^{1-\frac2p}r_i + \hw_i^{1-\frac2p}\right)^{-p} \left(d r_i^\frac{p}{p-2} + \hw_i\right)^{\frac{(p-2)^2}{2p}} \left|\left(\left(d \mR^\frac{p}{p-2} + \hmW\right)^{\frac12-\frac1q} \mU z\right)_i\right|^{p-2} (\mU z)_i^2 \nonumber \\
&\overset{(i)}{\le} \sum_{i \in [n]} \left(d^{1-\frac2p}r_i + \hw_i^{1-\frac2p}\right)^{-1} \hw_i^{-\frac{p-2}{2}}\left|\left(\left(d \mR^\frac{p}{p-2} + \hmW\right)^{\frac12-\frac1q} \mU z\right)_i\right|^{p-2} (\mU z)_i^2 \nonumber \\
&\overset{(ii)}{\le} \left\|\left(d \mR^\frac{p}{p-2} + \hmW\right)^{\frac12-\frac1q} \mU z\right\|_2^\frac{p-2}{2} \cdot \\ &\sum_{i \in [n]} \left(d^{1-\frac2p}r_i + \hw_i^{1-\frac2p}\right)^{-1} \hw_i^{-\frac{p-2}{2}} \sigma\left(\left(d \mR^\frac{p}{p-2} + \hmW\right)^{\frac12-\frac1q} \mU\right)_i^\frac{p-2}{2} (\mU z)_i^2 \\
&\overset{(iii)}{\le} 2^\frac{p-2}{2}\left\|\left(d \mR^\frac{p}{p-2} + \hmW\right)^{\frac12-\frac1q} \mU z\right\|_2^\frac{p-2}{2} \sum_{i \in [n]} \left(d^{1-\frac2p}r_i + \hw_i^{1-\frac2p}\right)^{-1} (\mU z)_i^2 \nonumber \\
&\overset{(iv)}{\le} 4^\frac{p-2}{2}\|\hy\|_2^p \le 4^\frac{p-2}{2}\left(3d^{1-\frac2p}\right)^\frac{p}{2} \le 4^p d^{\frac{p}{2}-1}. \label{eq:lastline}
\end{align}
Here, $(i)$ follows from the inequality $a^{1-2/p} + b^{1-2/p} \ge (a+b)^{1-2/p}$, which holds for all $a, b \ge 0$, for $a = d_ir_i^\frac{p}{p-2}$ and $b = \hw_i$, and the trivial bound \[ \left(d^{1-\frac2p}r_i + \hw_i^{1-\frac2p}\right)^{-p/2} \le \hw_i^{-\frac{p-2}{2}}. \] Also, $(ii)$ is shown using \cref{fact:lev}, and $(iii)$ uses the definition of $c$-regularized Lewis weights (\cref{def:reglewis}) for $c = dr^\frac{p}{p-2}$ as chosen in line \ref{line:approxlewissmall} of \OracleSmall~(\cref{algo:oraclesmall}) and that $\hw$ are $1.1$-approximate weights by \cref{lemma:reglevapprox}. Finally, $(iv)$ uses the definition of $\hy$ and
\[ (dr_i^\frac{p}{p-2} + \hw_i)^{1-\frac2q} = (dr_i^\frac{p}{p-2} + \hw_i)^{\frac2p-1} \le \max\left(d^{1-\frac2p}r_i, \hw_i^{1-\frac2p}\right)^{-1} \le 2\left(d^{1-\frac2p}r_i + \hw_i^{1-\frac2p}\right)^{-1}. \]

Taking $p$-th roots of \eqref{eq:lastline} shows that $\|y\|_p \le O(d^\frac{p-2}{2p-2})$ as desired.
To bound the cost, note that line \ref{line:approxlewissmall} and \ref{line:returny} of \OracleSmall~(\cref{algo:oraclesmall}) call $\O(1)$ solves to $\mU^\top\mD\mU$ for diagonal matrices $\mD$ by \cref{lemma:reglevapprox}.
\end{proof}

%% file: proofs.tex
\section{Lewis Weight Computation}
\label{sec:proofs}

\subsection{Proof of \cref{lemma:approxplarge}}
\label{proofs:approxplarge}
Our approach is based on \cite{CCLY19}, which approximates the John ellipse by performing $\O(1)$ iterations of the fixed point map for computing Lewis weights, and then averaging the results of the iterations. To speed up the iterations (as was done in previous work), we use that leverage scores can be efficiently found by sketching techniques \cite{SS11,DMMW12}. 
\begin{lemma}[Leverage score estimates \cite{SS11,DMMW12}]
\label{lemma:approxlev}
There is an algorithm $\ApproxLev(\mA, \eps)$ that for matrix $\mA \in \R^{n \times d}$ and $\eps \in (0, 1)$ returns a vector $w \in \R^n_{>0}$ satisfying $(1-\eps)w_i \le \sigma(\mA)_i \le (1+\eps)w_i$ for all $i \in [n]$. The algorithm succeeds whp. and runs in $\O(\eps^{-2} \cdot \T_\mA)$ time.
\end{lemma}

The algorithm that we use to show \cref{lemma:approxplarge} is given below.

\begin{algorithm}[ht]
\caption{$\ApproxLargeWeights(\mA, p).$
	 Given matrix $\mA \in \R^{n \times d}$ and $p \ge 2$, returns weights $w \in \R^n$ satisfying $\|w\|_1 \le 2d$ and $w \ge \sigma(\mW^{1/2-1/p}\mA).$ \label{algo:approxplarge}}
 \label{alg:approx_large}
$w_i^{(0)} \assign d/n$ for all $i \in [n]$. \\
\For{$k = 0, 1, \dots, T-1$}{
	$w^{(k+1)} \assign \ApproxLev((\mW^{(k)})^{1/2-1/p}\mA, 1/10)$ . \Comment{\cref{lemma:approxlev}}
}
\Return $w \assign \frac{3}{2T}\sum_{k=1}^T w^{(k)}.$
\end{algorithm}
 
Our analysis uses the convexity of a potential function on the inverse quadratic form.
\begin{lemma}
\label{lemma:potential}
For matrix $\mA \in \R^{n \times d}$ and $p \ge 2$, the function $\phi_i : \R_{\ge0}^n \to \R$ is convex, where
\[ \phi_i(v) = \log\left(v_i^{-2/p}a_i^\top \left(\mA^\top \diag{v}^{1-2/p} \mA\right)^{-1} a_i\right). \]
\end{lemma}
\begin{proof}
To start, we use \cite[Lemma 3.4]{CCLY19}, which says that the function
\begin{align} v \to \left(a_i^\top \left(\mA^\top \diag{v} \mA\right)^{-1} a_i\right) \label{eq:factabove} \end{align}
is convex. Note that the function $\log(v_i^{-2/p}) = -2 \log(v_i)/p$ is convex, hence it suffices to argue the convexity of the function
\[ v \to \log\left(a_i^\top \left(\mA^\top \diag{v}^{1-2/p} \mA\right)^{-1} a_i\right). \]
To show this, note that for any $u, v \in \R_{\ge0}^n$ and $\lambda \in [0, 1]$ we have
\begin{align*}
&\log\left(a_i^\top \left(\mA^\top \diag{\lambda u + (1-\lambda)v}^{1-2/p} \mA\right)^{-1} a_i\right) \\
&\le \log\left(a_i^\top \left(\mA^\top \left(\lambda\diag{u}^{1-2/p} + (1-\lambda)\diag{v}^{1-2/p}\right) \mA\right)^{-1} a_i\right) \\
&\le \lambda \log\left(a_i^\top \left(\mA^\top \diag{u}^{1-2/p} \mA\right)^{-1} a_i\right) + (1-\lambda) \log\left(a_i^\top \left(\mA^\top \diag{v}^{1-2/p} \mA\right)^{-1} a_i\right)
\end{align*}
as desired. The first inequality uses that the function $x \to x^{1-2/p}$ is concave for $p \ge 2$, and the final step uses the convexity of  \eqref{eq:factabove}.
\end{proof}

\begin{proof}[Proof of \cref{lemma:approxplarge}]
We show that for $T = 10 \log n$, \cref{algo:approxplarge} returns valid weights $w = \frac{3}{2T}\sum_{k=1}^T w^{(k)}.$ First, we show that $\|w\|_1 \le 2d.$ This follows because $\|w^{(k)}\|_1 \le 1.1d$ by \cref{lemma:approxlev}, so $\|w\|_1 \le 1.5 \cdot 1.1d \le 2d.$ Additionally \[ \|w\|_1 \ge \sum_{i \in [n]} \sigma_i(\mA^\top \mW^{1-\frac2p}\mA) = \mathrm{rank}(\mA) = d \] by our assumption that $\mA$ is full-rank.
Now we use \cref{lemma:potential} to get that
\begin{align*}
\phi_i\left(\frac{1}{T}\sum_{k \in [T]} w^{(k)}\right) &\le \frac{1}{T} \sum_{k \in [T]} \phi_i(w^{(k)}) = \frac{1}{T} \sum_{k \in [T]} \log\left(\sigma((\mW^{(k)})^{1/2-1/p}\mA)/w_i^{(k)} \right) \\ &\le \frac{1}{T} \sum_{k \in [T]} \left(\log(1.1) + \log(w^{(k+1)}_i / w^{(k)}_i) \right) \\
&\le 0.1 + \frac{1}{T}\log(w^{(T)}_i / w^{(0)}_i) \le 0.1 + \frac{1}{T} \log(1.1n/d) \le 0.2,
\end{align*}
where we have used \cref{lemma:approxlev}, the definition of $w^{(k+1)}$ in terms of $w^{(k)}$, and $w^{(T)}_i \le 1.1$ and the choice $w^{(0)}_i = d/n$ in the final line. Thus 
\[ w_i \ge 3/2 \cdot \exp(-0.2) \sigma\left(\mW^{1/2-1/p}\mA\right) \ge \sigma\left(\mW^{1/2-1/p}\mA\right). \]
To analyze the runtime, \cref{lemma:approxlev} shows that each of the $T$ iterations solves $O(\log n)$ systems in $\mA^\top (\mW^{(k)})^{1-2/p}\mA$, for a total of $O(\log^2 n)$ solves in $\mA^\top \mD \mA$ for diagonal matrices $\mD$.
\end{proof}

\subsection{Proof of \cref{lemma:reglevapprox}}
\label{proofs:reglevapprox}

The algorithm that shows \cref{lemma:reglevapprox} is given below. It is based on adapting the contractive map for computing Lewis weights in \cite{CP15} to regularized Lewis weights. Similar ideas have also appeared in \cite{BLLSSSW21}.
\begin{algorithm}[ht]
\caption{$\ApproxRegLewis(\mA, c, q).$ Given matrix $\mA \in \R^{n \times d}, c \in \R^n, q \le 2$, returns approximate $c$-regularized $\ell_q$ Lewis weights. \label{algo:reglevapprox}}
$w_i^{(0)} \assign 1$ for all $i \in [n]$. \\
\For{$k = 0, 1, \dots, T-1 = O(\log \log n)$}{
	$\sigma^{(k)} \assign \ApproxLev((\mC + \mW^{(k)})^{\frac12-\frac1q}\mA, 1/50)$. \Comment{\cref{lemma:approxlev}} \label{line:firstapprox}
	$w^{(k+1)} \assign \left(\left(\mC + \mW^{(k)}\right)^{\frac2q-1}\sigma^{(k)} + \left(\mC + \mW^{(k)}\right)^{\frac2q-1}c\right)^{q/2} - c.$ \label{line:contract}
}
\Return $\ApproxLev((\mC + \mW^{(T)})^{\frac12-\frac1q}\mA, 1/50).$ \label{line:finalapprox}
\end{algorithm}
The key point is that the operation in line \ref{line:contract} of \ApproxRegLewis~(\cref{algo:reglevapprox}) is a contractive map.
For the purposes of the lemma and proof below, we say that $a \approx_\nu b$ if $\nu^{-1}a \le b \le \nu b.$ This way $a \approx_1 b$ means that $a = b$.
\begin{lemma}[Contractive map]
\label{lemma:contract}
For matrix $\mA \in \R^{n \times d}, c \in \R^n, q \le 2$, let $w$ be the $c$-regularized $\ell_q$ Lewis weights. For $u \in \R_{\ge0}^n$, let $\mU = \diag{u}, \mC = \diag{c}$, $\sigma \approx_{51/50} \sigma\left((\mC + \mU)^{\frac12-\frac1q}\mA \right)$, and
\[ u^\new = \left(\left(\mC + \mU\right)^{\frac2q-1}\sigma + \left(\mC + \mU\right)^{\frac2q-1}c\right)^{q/2} - c. \]
If $\mU+\mC \approx_{\nu} \mW+\mC$ for some $\nu \ge 1$, then $\mC+\mU^\new \approx_{(51/50)^{q/2} \nu^{1-\frac{q}{2}}} \mC+\mW.$
\end{lemma}
\begin{proof}
Note that \[ \left[\left(\mC + \mU\right)^{\frac2q-1}\sigma\left((\mC + \mU)^{\frac12-\frac1q}\mA \right)\right]_i = a_i^\top\left(\mA^\top(\mC + \mU)^{1-\frac2q}\mA\right)^{-1}a_i. \]
Because $\mC + \mU \approx_{\nu} \mC + \mW$ and $\sigma \approx_{51/50} \sigma\left((\mC + \mU)^{\frac12-\frac1q}\mA\right)$ we get that
\[ \left(\mC + \mU\right)^{\frac2q-1}\sigma \approx_{51/50 \nu^{\frac2q-1}} \left(\mC + \mW\right)^{\frac2q-1}\sigma\left((\mC + \mW)^{\frac12-\frac1q}\mA \right) = \left(\mC + \mW\right)^{\frac2q-1}w. \]
Similarly, we know that $\left(\mC + \mU\right)^{\frac2q-1}c \approx_{\nu^{\frac2q-1}} \left(\mC + \mW\right)^{\frac2q-1}c.$ Summing and raising to the $q/2$ power we get that
\begin{align*}
u^\new + c &= \left(\left(\mC + \mU\right)^{\frac2q-1}\sigma + \left(\mC + \mU\right)^{\frac2q-1}c\right)^{q/2} \\
&\approx_{(51/50\nu^{\frac2q-1})^{q/2}} \left(\left(\mC + \mU\right)^{\frac2q-1}w + \left(\mC + \mW\right)^{\frac2q-1}c\right)^{q/2}
= w + c.
\end{align*}
This completes the proof.
\end{proof}
We can use this to show \cref{lemma:reglevapprox}.
\begin{proof}[Proof of \cref{lemma:reglevapprox}]
By the proof of \cite[Lemma 3.5]{CP15}, we have that $w^{(0)} + c \approx_{\nu_0} w+c$ for some $\nu_0 \le O(n^{1-q/2})$. Applying \cref{lemma:contract} to $u = w^{(k)}$ for $k = 0, \dots, T-1$ for $T = O(\log \log n)$ shows by induction that $c+w^{(T)} \approx_{51/50\nu_0^{(1-q/2)^T}} c+w$, so $c+w^{(T)} \approx_{41/40} c+w$ for $T = O(\log \log n)$ and any $q \ge 1/2$. Because
\[ \left[\sigma\left(\left(\mC + \mW^{(T)}\right)^{\frac12-\frac1q}\mA \right)\right]_i = (c_i + w_i^{(T)})^{1-\frac2q}a_i^\top\left(\mA^\top(\mC + \mW^{(T)})^{1-\frac2q}\mA\right)^{-1}a_i \approx_{(41/40)^{\frac4q-2}} w. \]
To finish, note that the algorithm returns a $51/50$ approximate leverage score in line \ref{line:finalapprox} of \ApproxRegLewis~(\cref{lemma:reglevapprox}), so the total approximation is at most $(41/40)^{\frac4q-2} \cdot 51/50 \le 1.1$.

To bound the computation cost, note that the only calls to solvers for $\mA^\top\mD\mA$ are in the approximate leverage score computation in lines \ref{line:firstapprox} and \ref{line:finalapprox} of \ApproxRegLewis~(\cref{lemma:reglevapprox}). Each call requires $\O(1)$ solves by \cref{lemma:approxlev}, and it is calls $T+1 = O(\log\log n)$ times, as desired.
\end{proof}

\section{Additional Proofs}
\label{proofs:cor12}

In this section we show the numerical inequalities in \cref{lemma:cor1,lemma:cor2}.

\begin{proof}[Proof of \cref{lemma:cor1}]
By \cref{lemma:refine} we get that $(a+b)^k - a^k \le ka^{k-1}b + 2k^2a^{k-2}b^2 + k^k b^k.$ Combining this with Young's inequality on the middle term $2k^2a^{k-2}b^2$, i.e.
\[ 2k^2 a^{k-1} b \leq 2k\left(\frac{k - 2}{k - 1} a^{k - 1} +  \frac{1}{k - 1} (kb)^{k - 1}\right) \leq 2ka^{k - 1}b + 2k^kb^k \] completes the proof.
\end{proof}

\begin{proof}[Proof of \cref{lemma:cor2}]
It suffices to show $(1+x)^k - 1 \le 4^k(x + x^k)$ for $x = b/a.$ If $x \le 1$ then note
\[ (1+x)^k - 1 = \int_0^x k(1+t)^{k-1} dt \le x \cdot \sup_{t \in [0, 1]} k(1+t)^{k-1} \le 2^{k-1}kx \le 4^k x \] for $k \ge 1$.
If $x \ge 1$ then $(1+x)^k \le (2x)^k \le 2^k x^k$ as desired.
\end{proof}

%% file: linf.tex
\section{Lewis Weights for $\ell_\infty$ Regression}
\label{sec:appendixinf}

Here, we argue that using $\ell_\infty$ Lewis weight overestimates (\cref{def:overestimate}) along with the computations and framework of \cite{CJJJLST20} directly give an algorithm for $\ell_\infty$-regression, proving \cref{thm:linfmain}. We use the notion of \emph{quasi-self-concordance} from \cite[Definition 10]{CJJJLST20}.

\begin{lemma}
\label{lemma:qsc}
Define $\lse_t: \R^n \to \R$ as $\lse_t(x) = t \log\left(\sum_{i\in[n]} \exp(x_i/t) \right).$ If $w \in \R^n_{\ge0}$ are $\ell_\infty$ Lewis weight overestimates (\cref{def:overestimate}) then $\lse_t(x)$ is $1/t$-smooth and $2/t$-quasi-self-concordant in the norm $\| \cdot \|_{\mA^\top \mW \mA}.$
\end{lemma}
\begin{proof}
By scaling, it suffices to show the result for $t = 1$. The computations in the proof of \cite[Lemma 14]{CJJJLST20} in \cite[Appendix G.1]{CJJJLST20} shows that
$u^\top \g^2 \lse(x) u \le \|u\|_\infty^2.$ Now by \cref{lemma:weightkey} for $p = \infty$ we get $\|u\|_\infty \le \|u\|_{\mA^\top \mW \mA}$ which implies the smoothness claim.

To show quasi-self-concordance, we use the computations in the proof of \cite[Lemma 14]{CJJJLST20} in \cite[Appendix G.1]{CJJJLST20} to get
\[ \left|\g^3\lse(x)[u, u, h]\right| \le \|u\|_{\g^2\lse(x)}^2 \|h\|_\infty. \] 
The proof follows from the fact that $\|h\|_\infty \le \|h\|_{\mA^\top \mW \mA}$ by \cref{lemma:weightkey} with $p = \infty$.
\end{proof}

We can plug this bound into \cite[Corollary 12]{CJJJLST20} to show \cref{thm:linfmain}.
\begin{proof}[Proof of \cref{thm:linfmain}]
Define $\xopt = \argmin_x \|\mA x - b\|_\infty$ and
 $\OPT = \|\mA \xopt - b\|_\infty.$ We assume that we start at a point $x \in \R^n$ with $\|\mA x - b\|_\infty \le 2\OPT$. Otherwise, the same proof shows that given any upper bound on $\OPT$, the algorithm allows us to reduce the error by a constant factor. We can initialize with polynomial error by solving the $\ell_2$ problem $\min_{x \in \R^d} \|\mA x - b\|_\infty.$ We set $t = \frac{\eps \OPT}{20 \log n}$ (which loses $\eps \OPT/2$ additive function accuracy) and minimize $\lse_t(\mA x - b)$ to $\eps \OPT /2$ accuracy. In \cite[Corollary 12]{CJJJLST20} we set $\mM = \mA^\top \mW \mA$ and $M = 2/t = O\left(\frac{20 \log n}{\eps \OPT} \right)$ by \cref{lemma:qsc}.

We show we can set $R = O(\OPT \sqrt{d}).$ Indeed note that $\|\mA(x-x^*)\|_\infty \le \|\mA x - b\|_\infty + \|\mA \xopt - b\|_\infty \le 3 \OPT.$ Thus we get
\[ R^2 = \|x - \xopt\|_{\mA^\top \mW \mA}^2 \le \sum_{i\in[n]} w_i (\mA(x-\xopt))_i^2 \le 9 \OPT^2 \sum_{i\in[n]} w_i \le 18d \OPT^2 \] as $\|w\|_1 \le 2d$ by the construction in \cref{lemma:approxplarge}. Pluggin in this value of $R, M$ into \cite[Corollary 12]{CJJJLST20} gives an iteration bound of
\[ \O((RM)^{2/3}) = \O\left(\left(\OPT \sqrt{d} \cdot \frac{20 \log n}{\eps \OPT} \right)^{2/3} \right) = \O(d^{1/3}\eps^{-2/3}). \]
\end{proof}